\documentclass[twoside]{article}

\usepackage[accepted]{aistats2026}
\usepackage[round]{natbib}

\usepackage[utf8]{inputenc} 
\usepackage[T1]{fontenc}    
\usepackage{hyperref}       
\usepackage{url}            
\usepackage{booktabs}       
\usepackage{amsfonts}       
\usepackage{nicefrac}       
\usepackage{microtype}      
\usepackage{xcolor}         
\usepackage{amsmath, amssymb, amsthm}
\usepackage{algorithm}
\usepackage{mathtools}
\usepackage{float}
\usepackage{framed}
\usepackage[noend]{algpseudocode}
\usepackage{enumitem}
\usepackage{thmtools, thm-restate}
\usepackage{wrapfig}
\usepackage{tabularx}
\usepackage{makecell}
\usepackage{booktabs} 
\usepackage{etoolbox}
\usepackage{graphicx}
\usepackage{subcaption}
\usepackage{multirow}
\usepackage{multicol}
\usepackage[capitalize,noabbrev]{cleveref}
\usepackage{placeins}
\usepackage{balance}

\newcommand{\ra}{\tau}
\newcommand{\ma}{\gamma}
\newcommand{\Den}{\textup{\texttt{Decomp}}}
\newcommand{\De}[3]{\Den\left(#1,#2,#3\right)}
\newcommand{\Degt}[1]{\De{#1}{\ma}{\ra}}
\newcommand{\Recn}{\textup{\texttt{Recon}}}
\newcommand{\Rec}[3]{\Recn\left(#1,#2,#3\right)}
\newcommand{\Rectg}[1]{\Rec{#1}{\ra}{\ma}}
\newcommand{\measure}{\texttt{MeasureResidual} }
\renewcommand{\vec}{\text{vec}}
\newcommand{\ten}{\text{ten}}
\newcommand{\p}{\mathbf{p}}
\newcommand{\A}{\mathcal{A}}

\newcommand{\M}{\mathcal{M}}
\newcommand{\N}{\mathcal{N}}

\newcommand{\R}{\mathbb{R}}

\newcommand{\X}{\mathcal{X}}
\newcommand{\D}{\mathcal{D}}
\newcommand{\I}{\mathbb{I}}
\newcommand{\Y}{\mathcal{Y}}
\newcommand{\W}{\mathcal{W}}
\newcommand{\bigO}{\mathcal{O}}
\newcommand{\norm}[1]{\left\lVert#1\right\rVert}

\newcommand{\abs}[1]{|#1|}
\newcommand{\paren}[1]{\left(#1\right)}
\newcommand{\set}[1]{\left\{#1\right\}}

\DeclareMathOperator*{\argmin}{argmin}

\newcommand{\init}{\text{init}}
\newcommand{\used}{\text{used}}

\theoremstyle{plain}
\newtheorem{theorem}{Theorem}[section]

\newtheorem{proposition}[theorem]{Proposition}

\theoremstyle{definition}
\newtheorem{definition}[theorem]{Definition}

\theoremstyle{remark}

\begin{document}

%

%
\runningauthor{Fuentes, Mullins, Xiao, Kifer, Musco, Sheldon}

\twocolumn[

\aistatstitle{Fast Private Adaptive Query Answering for Large Data Domains}

\aistatsauthor{ Miguel Fuentes*\\University of Massachusetts Amherst \And Brett Mullins*\\University of Massachusetts Amherst \AND  Yingtai Xiao\\TikTok \And Daniel Kifer\\Penn State University \AND Cameron Musco\\University of Massachusetts Amherst \And Daniel Sheldon\\University of Massachusetts Amherst }

\aistatsaddress{} ]

\begin{abstract}
Privately releasing marginals of a tabular dataset is a foundational problem in differential privacy. 
However, state-of-the-art mechanisms suffer from a computational bottleneck when marginal estimates are reconstructed from noisy measurements.
Recently, \emph{residual queries} were introduced and shown to lead to highly efficient reconstruction in the batch query answering setting.
We introduce new techniques to integrate residual queries into state-of-the-art adaptive mechanisms such as AIM.
Our contributions include a novel conceptual framework for residual queries using multi-dimensional arrays, lazy updating strategies, and adaptive optimization of the per-round privacy budget allocation.
Together these contributions reduce error, improve speed, and simplify residual query operations.
We integrate these innovations into a new mechanism (AIM+GReM), which improves AIM by using fast residual-based reconstruction instead of a graphical model approach. 
Our mechanism is orders of magnitude faster than the original framework and demonstrates competitive error and greatly improved scalability.

\end{abstract}

\section{Introduction}

Releasing marginal statistics under differential privacy (DP) is vital for sharing data insights while protecting privacy. 
Marginals (frequency distributions over attribute subsets) are crucial for descriptive statistics, downstream analyses, and various applications. 
We consider the problem of answering a workload $\W$ of marginal queries with a mechanism that satisfies differential privacy, where each $\ma \in \W$ is a set of attributes for which we want to estimate a marginal $\mu_\gamma$ (an array of counts for each possible setting of those attributes). 
The most common strategy for this problem is the ``select-measure-reconstruct'' paradigm. 
Instead of naively measuring every query in $\W$ via a simple mechanism such as noise addition, this approach involves three steps: first, the mechanism intelligently \emph{selects} a query or set of queries; second, it privately \emph{measures} these selected queries; and third, it uses these measurements to \emph{reconstruct} estimates for all queries in the original workload. 
We focus on \emph{adaptive} or \emph{data-dependent} approaches, which iteratively select and measure marginals.
MWEM~\citep{mwem} was an influential early mechanism of this type and has been followed by a family of successful mechanisms \citep{aim, JAM, liu2021iterative, liu2021leveraging}. 
Adaptive approaches typically give lower error than batch approaches such as the matrix mechanism family of algorithms~\citep{matrixmech, hdmm, resplanner}, especially with a very restrictive privacy budget, because they can save budget by only measuring poorly approximated marginals.

AIM \citep{aim} is a state-of-the-art select-measure-reconstruct method for marginal queries and synthetic data.
However, its performance is limited on large data domains by the computational tractability of Private-PGM \citep{privpgm}, its reconstruction method.
Private-PGM solves a principled optimization problem to reconstruct a graphical model representation of the data distribution, but for some sets of measured marginals the representation becomes computationally intractable. 
AIM therefore employs a heuristic during selection to avoid measurements that would cause the model size to exceed a specified limit.
Even with this safeguard, the running time of AIM is prohibitive in some scenarios.

Recent work has introduced a promising new reconstruction approach based on \emph{residual queries}~\citep{resplanner,grem}. 
Residuals are statistical queries that can be viewed as decompositions of marginals into tables with non-overlapping information. 
The GReM (Gaussian residuals-to-marginals) algorithm~\citep{grem} shows how to convert marginal measurements with Gaussian noise to equivalent residual measurements and then reconstruct workload marginals from the noisy residuals via a principled optimization problem.
Unlike Private-PGM, this reconstruction is always efficient. 
Therefore, it is very promising for use in adaptive mechanisms.

Our work advances this line of research with several key contributions.
First, we introduce a novel \emph{in-axis formulation} that uses elementary array operations to define and convert between marginals and residuals. This approach clarifies these operations, simplifies the complexity analysis, and speeds marginal decomposition.
Second, we contribute \emph{lazy updating}, a procedure significantly speeding up GReM reconstruction in the context of an iterative mechanism.
Third, we develop \emph{conditional ResidualPlanner} to optimize noise allocation during measurement, named after the ResidualPlanner method which optimizes noise allocation in the batch setting~\citep{resplanner}.
Finally, we integrate these ideas into AIM to form AIM+GReM, an end-to-end adaptive mechanism that measures residuals and uses GReM for reconstruction.
AIM+GReM is orders of magnitude faster than AIM and achieves favorable time-vs-utility tradeoffs compared to AIM and ResidualPlanner, the optimal batch method. 
Additionally, conditional ResidualPlanner provides an optimal solution for the problem of noise allocation under pre-existing measurements, a result of independent interest for any adaptive or multi-stage measurement strategy.

\section{Preliminaries} \label{s:preliminaries}

Let $\D$ be a discrete tabular dataset consisting of records $x^{(1)}, \ldots, x^{(N)}$.
Each record $x = \paren{x_1, \ldots, x_d}$ contains $d$ categorical attributes.
The $i$th attribute $x_i$ belongs to the finite set $\X_i = \set{1, \ldots, n_i }$.
The data universe is $\X = \prod_{i=1}^d \X_i$ and has size $n = \prod_i n_i$.
We represent a dataset $\D$ as a multidimensional array $\A \in \R^{n_1 \times \cdots \times n_d}$ with one dimension for each attribute, where $\A[i_1, \ldots, i_d] = \sum_{j=1}^N \prod_{k=1}^{d} \I[x_{k}^{(j)} = i_k ]$.
In prior work, the dataset $\D$ was represented as a vector $p \in \R^n$, called the \emph{data vector}, obtained by flattening the array $\A$ into an $n \times 1$ vector \citep{mwem, hardt2012simple, matrixmech, hdmm, resplanner, grem}.

\textbf{Marginals}.
Marginals are a common type of statistical query that describe the frequency or count of combinations of values across a subset of attributes. 
For example, on a demographics dataset, the marginal over age group and education level counts the number of individuals for each combination of age group (e.g., 18-25, 26-35, 36-50, 50+) and level of education (e.g. no college, some college, obtained degree).
For a subset of attributes $\ma \subseteq [d]$, the marginal over $\ma$, given by $\mu_\ma$, is an array with $|\ma|$ dimensions, where $\mu_\ma[i_1, \ldots, i_{|\ma|}] = \sum_{j=1}^N \prod_{k \in \ma} \I[x_{k}^{(j)} = i_k]$.
Figure \ref{fig:marginal_age_educ} shows an example marginal over age and education level, while Figure\ref{fig:marginal_age} shows the marginal over age obtained by summing the values of education level from the two way marginal.

\begin{figure}[ht]
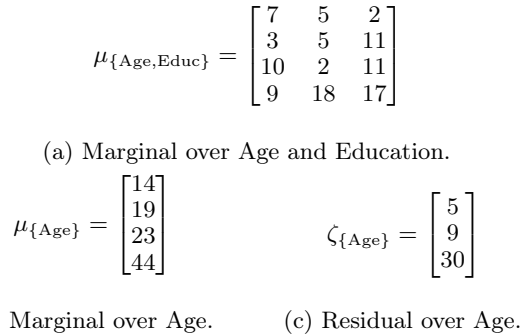

  \centering
  \begin{subfigure}[b]{0.40\textwidth}
    \small
    $$
      \mu_{\{\text{Age}, \text{Educ}\}} = 
      \begin{bmatrix}
          7   & 5  & 2  \\
          3   & 5  & 11 \\
          10  & 2  & 11 \\
          9   & 18 & 17
      \end{bmatrix}
    $$
    \caption{\centering Marginal over Age and Education.}
    \label{fig:marginal_age_educ}
  \end{subfigure} 
  \hfill
  \begin{subfigure}[b]{0.24\textwidth}
    \small
    $$
      \mu_{\{\text{Age}\}} = 
      \begin{bmatrix}
          14  \\
          19 \\
          23 \\
          44
      \end{bmatrix}
    $$
    \caption{\centering Marginal over Age.}
    \label{fig:marginal_age}
  \end{subfigure}
  \hfill
  \begin{subfigure}[b]{0.24\textwidth}
    \small
    $$
      \zeta_{\{\text{Age}\}} = 
      \begin{bmatrix}
          5  \\
          9 \\
          30
      \end{bmatrix}
    $$
    \caption{\centering Residual over Age.}
    \label{fig:res_age}
  \end{subfigure}
  \caption{Example Marginals and Residual.}
  \label{fig:example_marginal}  
\end{figure}

\textbf{Differential privacy}.
Differential privacy is a framework that ensures the privacy of individuals in a dataset by limiting the influence of any single record on the output of a computation. 
This approach is motivated by the need to balance data utility with privacy, enabling the safe analysis of sensitive information.
\begin{definition}(Zero Concentrated Differential Privacy; \citet{zcdp})
  Let $\M: \X \rightarrow \Y$ be a randomized mechanism. For any datasets $\D, \D'$ that differ by at most one record, $\M$ satisfies $\rho$-zero concentrated differential privacy, denoted $\rho$-zCDP, if
  \begin{equation}
      D_\alpha\paren{\M\paren{\D}\|\M\paren{\D'}} \leq \alpha\rho
  \end{equation}
  for all $\alpha \in \paren{1, \infty}$, where $D_\alpha$ denotes the $\alpha$-Renyi divergence.
\end{definition}

This definition of differential privacy is convenient for our setting and can be converted to the more common $\paren{\epsilon,\delta}$-approximate differential privacy \citep{dwork2006calibrating}. This conversion and other privacy properties are given in Appendix \cref{app:dp}.

\section{Residuals and Marginals as Multi-Dimensional Arrays} \label{sec:residuals-and-marginals}

In this section, we introduce residuals via a new and simplified framework. Traditionally, residuals are defined as transformations of marginals via query matrices constructed explicitly using Kronecker products and other algebraic techniques. Our approach departs from this vector-based approach, instead  defining the same transformations via simple and efficient ``in-axis'' transformations of multi-dimensional arrays. Due to the complexity involved in the previous approach, here we present only the new view. In Appendix \cref{app:res_props}, we give background on the previous approaches and demonstrate that our construction is equivalent to the methods in \citet{resplanner} and \citet{grem}.

\textbf{Residuals from marginals.} 
The residual $\zeta_\tau$ over attributes $\ra \subseteq [d]$ is an array with $|\ra|$ dimensions defined using \cref{alg:decomp} as $\zeta_\ra = \De{\mu_\tau}{\tau}{\tau}$, where $\mu_\tau$ is the $\tau$ marginal.\footnote{It is also true that $\zeta_\ra = \De{\mu_\ma}{\ma}{\tau}$ for any $\ma \supseteq \ra$.}
This applies a subtraction operator (details below) to each dimension of the $\tau$ marginal, leading to an array of size one smaller in each dimension.
For example, $\zeta_{\set{\text{Age}}}$ in \cref{fig:example_marginal} is obtained by subtracting the first entry of $\mu_{\set{\text{Age}}}$ from each of the last three entries.
We can call $\De{\mu_\ma}{\ma}{\ra}$ to \emph{decompose} $\mu_\ma$ into component residuals $\zeta_\tau$ for $\tau \subseteq \gamma$.
The \Recn{} procedure (\cref{alg:recon}) then \emph{reconstructs} a marginal from component residuals. 
A marginal and its component residuals express the same information in the sense that \Den{} and \Recn{} provide an invertible linear transformation between them:

\begin{proposition} \textup{\citep{grem}} \label{prop:bijection}
  The function $T_\ma\paren{\mu_\ma} = \set{\Degt{\mu_\ma} | \ra\subseteq\ma}$ is an invertible linear transformation between a marginal $\mu_\ma$ and a set of residuals $\zeta_\ma = \{ \zeta_{\ra} \mid \ra \subseteq \ma \}$. The inverse is given by $T_\ma^{-1}\paren{\zeta_\ma} = \sum_{\ra \subseteq \ma} \Rectg{\zeta_\ra}$.
\end{proposition}

Note that this bijection is between a single marginal, $\mu_\ma$, and the set of all its component residuals, $\{ \zeta_{\ra} \mid \ra \subseteq \ma \}$. A single marginal cannot be recovered from just one of its component residuals. A worked example of this process is provided in Appendix~\cref{app:example} along with related proofs.

\textbf{Benefits of residuals}
While one marginal and its component residuals contain the same information, we can use residuals to efficiently perform inference tasks involving noisy marginals that are not tractable without residual decomposition. This process of inference from noisy marginals via residuals is covered in \cref{sec:grem}. Additionally, the running time of \Den{} and \Recn{} is nearly linear in the size $n_\gamma$ of $\mu_\gamma$, the larger of the input and output arrays, so it is tractable to use residuals as an intermediate representation:

\begin{theorem} \label{thm:time_complexity}
    Let $\mu_\gamma$ be a marginal over attributes $\gamma$ and $\zeta_\tau$ be a residual over attributes $\tau$ such that $\tau \subseteq \gamma$. Then \textup{Decomp}$(\mu_\gamma, \gamma, \tau)$ and \textup{Recon}$(\zeta_\tau, \tau, \gamma)$ take $\bigO \big( |\gamma| n_\gamma \big)$ time, where $n_\gamma = \prod_{i \in \gamma} n_i$.
\end{theorem}

\begin{figure*}[t] 
    \begin{minipage}[t]{0.48\textwidth}
        \begin{algorithm}[H]
          \caption{$\Degt{\mu}$}
          \label{alg:decomp}
          \begin{algorithmic}[1]
            \Require Marginal $\mu$, marginal attributes $\gamma$, residual attributes $\tau \subseteq \gamma$
            \Ensure Residual $\zeta$ for attributes $\tau$
            \State $\zeta \gets \mu$
            \For{$i \in \gamma \setminus \tau$}
              \State $\zeta \gets \texttt{sum}(\zeta, \texttt{axis}=i)$
            \EndFor
            \For{$j \in \tau$}
              \State $\zeta \gets \texttt{sub}(\zeta, \texttt{axis}=j)$
            \EndFor
            \State \Return $\zeta$
          \end{algorithmic}
        \end{algorithm}
    \end{minipage}
    \hfill 
    \begin{minipage}[t]{0.48\textwidth}
        \begin{algorithm}[H]
          \caption{$\Rectg{\zeta}$}
          \label{alg:recon}
          \begin{algorithmic}[1]
            \Require Residual $\zeta$, residual attributes $\tau$, marginal attributes $\gamma$, axis lengths $\{n_i\}_{i \in \gamma \setminus \tau}$
            \Ensure Marginal component $\mu_{\gamma,\tau}$ for attributes $\ma$
            \State $\mu_{\gamma,\tau} \gets \zeta$
            \For{$j \in \tau$}
              \State $\mu_{\gamma,\tau} \gets \texttt{center}(\mu_{\gamma,\tau}, \texttt{axis}=j)$
            \EndFor
            \For{$i \in \gamma \setminus \tau$}
              \State $\mu_{\gamma,\tau} \gets \texttt{smear}(\mu_{\gamma,\tau}, \texttt{axis}=i, \texttt{k}=n_i)$
            \EndFor
            \State \Return $\mu_{\ma, \gamma}$
          \end{algorithmic}
        \end{algorithm}
    \end{minipage}
    \caption{The decomposition (left, Algorithm~\ref{alg:decomp}) and reconstruction (right, Algorithm~\ref{alg:recon}) procedures.}
    \label{fig:main_algorithms}
\end{figure*}

\subsection{Transformations via In-Axis Operations}

The primary advantage of our framework is its structural approach to the transformations used to define and compute with residuals. Instead of defining them via a complicated algebraic construction, we define them as a sequence of simple operations applied to the axes of a multi-dimensional array. This decomposition of a complex transformation into ``in-axis’’ steps simplifies the conceptual model and its implementation. Specifically, each ``in-axis'' operation is a linear transformation applied to all 1D vectors along a specified axis of the array. For example, \texttt{sub}(\texttt{z}, \texttt{axis} = i) applies a differencing operation to every vector along axis i. The operations underpinning our \Den{} and \Recn{} procedures are detailed in \cref{tab:axis_ops}.

\begin{table}[ht]
    \caption{Along-dimension operations. (Note that \texttt{smear} is always applied to a singleton axis, making \texttt{v} a scalar.)}
    \label{tab:axis_ops}
    \centering
    \begin{tabularx}{\linewidth}{@{} l X @{}}
        \toprule
        \textbf{Operation} & \textbf{Implementation} \\
        \midrule
        \texttt{\_sum(v)} & \texttt{sum(v)} \\
        \texttt{\_sub(v)} & \texttt{v[1:]--v[0]} \\
        \texttt{\_center(v)} & u = \texttt{[0, v]} \\
                         & \texttt{return u-mean(u)} \\
        \texttt{\_smear(v, k)} & \texttt{v/k * ones(k)} \\
        \bottomrule
    \end{tabularx}
\end{table}

These operations yield the following behavior when applied to the full array: 
\texttt{sum} marginalizes the specified axis to size one;
\texttt{sub} reduces the axis size by one via differencing against the first slice to capture variations;
\texttt{center} restores the axis dimension lost by \texttt{sub} by inserting an initial slice of zeros then subtracting the mean across the specified axis,
and \texttt{smear} expands a unit axis to length \texttt{k} through scaling by $1/\texttt{k}$ and broadcasting (so the sum across the axis is unchanged). These operations can be implemented efficiently with multi-dimensional array libraries, using built-in operations much faster than mapping individual transformations across all vectors along an axis.

\section{Algorithmic Improvements for Iterative Mechanisms}
Residuals are very promising for use in adaptive mechanisms. 
The GReM (Gaussian residuals-to-marginals) framework~\citep{grem} provides techniques for measuring noisy residuals and using them to reconstruct marginals. This section describes GReM and introduces algorithmic improvements to increase the speed and accuracy of measuring and reconstructing in an adaptive mechanism.

\subsection{GReM}
\label{sec:grem}
GReM is a method for reconstructing a set of marginals $\W$ from a collection of noisy marginal or residual measurements.

The measurement model assumes that noisy residuals are obtained by first adding isotropic Gaussian noise to a marginal and then applying the $\Den{}$ transformation. This procedure yields a noisy residual $z_\ra$ distributed as $z_\ra = \zeta_\ra + \mathcal{N}(0, \sigma_\ra^2 V_\ra)$, where $\zeta_\ra$ is the true residual and $V_\ra$ is a covariance matrix derived from the $\Den{}$ transformation \citep{resplanner,grem}.\footnote{We interpret multidimensional arrays such as $z_\ra$, $\zeta_\ra$ as vectors in expressions for Gaussian distributions.}

The reconstruction is performed by the GReM-MLE algorithm, which is detailed in ~\cref{alg:grem_mle}. The procedure has two main stages. First, it computes a consolidated estimate for every residual $\tau$ in the downward closure of the target workload, $\W^\downarrow = \{ \ra \subseteq \ma \mid \ma \in \W \}$. This is done by calculating the inverse-variance weighted average of all available measurements for that residual. If no measurements for $\tau$ exist, its estimate is zero. Second, the algorithm synthesizes each target marginal $\hat{\mu}_\gamma$ for $\gamma \in \W$ by summing the reconstructed components derived from the residual estimates $\{z_\tau \mid \tau \subseteq \gamma\}$.

\begin{algorithm}[ht]
    \caption{GReM-MLE Reconstruction~\citep{grem}}
    \label{alg:grem_mle}  
    \begin{algorithmic}[1]
        \Require Workload $\W$, noisy residuals $z_{\tau, i} = \zeta_\tau + \mathcal N(0, \sigma_{\tau, i}^2 V_\tau)$ for $\tau \in \mathcal W^{\downarrow}$, $i=1, \ldots k_\tau$
        \Ensure Marginal estimates $\hat \mu_\gamma$ for each $\gamma \in \W$
        \State 
        $z_\tau = 
        \begin{cases}
            \frac{\sum_{i=1}^{k_i} z_{\tau,i} \sigma_{\tau,i}^{-2}}{\sum_{i=1}^{k_i}\sigma_{\tau,i}^{-2}} & k_\tau > 0 \\ 
            0 & k_\tau = 0
        \end{cases}$ for $\tau \in \W^{\downarrow}$
        \State $\hat \mu_\gamma = \sum_{\tau \subseteq \gamma} \Rec{z_\tau}{\tau}{\gamma}$ for $\gamma \in \W$
    \end{algorithmic}
\end{algorithm}

As formalized in Theorem~\ref{thm:grem_mle}, this two-stage procedure is guaranteed to produce estimates that are both consistent (i.e., there exists a data array $\hat{\A}$ that realizes all estimated marginals $\hat{\mu}_\gamma$ and residuals $z_\tau$) and optimal in that they maximize the joint log-likelihood of the initial measurements.

\begin{theorem}\textup{\citep{grem}} \label{thm:grem_mle}
 GReM-MLE computes marginals $\hat \mu_\gamma$ and residuals $z_\tau$ that are consistent with a data array $\hat \A$ and, subject to this constraint, minimize the weighted squared error $\sum_{\tau \in \W^{\downarrow}} \sum_{i=1}^{k_\tau}\frac{1}{2\sigma_{\tau,i}^2}\|z_\tau - z_{\tau,i}\|_{V_\tau}^2$, where $\|u\|^2_{V_\tau} = u^\top V_\tau^{-1} u$.
\end{theorem}

\subsection{Lazy Updating}

When GReM-MLE is used within an iterative algorithm (e.g., one that acquires new measurements sequentially), its computations can be significantly accelerated. After each new measurement, a full reconstruction of all target marginals is inefficient and unnecessary.

Specifically, when a new measurement for a residual $\alpha$ is obtained, only the corresponding consolidated estimate $z_\alpha$ is affected by the inverse-variance weighted average update. All other residual estimates $\{z_\tau\}_{\tau \neq \alpha}$ remain unchanged. Due to the linear nature of the reconstruction step, the impact of this single update on the marginals is localized and can be computed incrementally, as formalized below.

\begin{proposition} \label{prop:lazy_update}
Suppose an algorithm maintains residual estimates $\{z_\tau\}_{\tau \in \W^\downarrow}$ and marginal estimates $\hat{\mu}_\gamma = \sum_{\tau \subseteq \gamma} \texttt{Recon}(z_\tau, \tau, \gamma)$ for all $\gamma \in \W$. If a single residual estimate $z_\alpha$ is updated to a new value $z'_\alpha$, the new marginal estimates $\hat{\mu}'_\gamma$ are:
$$
\hat \mu'_\gamma = \begin{cases}
\hat \mu_\gamma & \alpha \not\subseteq \gamma \\
\hat \mu_\gamma + \texttt{Recon}(z_\alpha' - z_\alpha, \alpha, \gamma) & \alpha \subseteq \gamma.
\end{cases}
$$
\end{proposition}
\begin{proof}
If $\alpha \not\subseteq \gamma$, the set of residuals summed to form $\hat{\mu}_\gamma$ is unchanged, so $\hat{\mu}'_\gamma = \hat{\mu}_\gamma$. Otherwise, if $\alpha \subseteq \gamma$:
$$
\begin{aligned}
\hat \mu'_\gamma &= \sum_{\tau \subseteq \gamma, \tau \neq \alpha} \texttt{Recon}(z_\tau, \tau, \gamma) + \texttt{Recon}(z'_\alpha, \alpha, \gamma) \\
&= \left( \hat \mu_\gamma - \texttt{Recon}(z_\alpha, \alpha, \gamma) \right) + \texttt{Recon}(z'_\alpha, \alpha, \gamma) \\
&= \hat \mu_\gamma + \texttt{Recon}(z'_\alpha - z_\alpha, \alpha, \gamma).
\end{aligned}
$$
The final step holds due to the linearity of the $\texttt{Recon}$ operator with respect to its first argument.
\end{proof}

This result demonstrates a key principle for efficient implementation: by storing the full set of residual estimates, an update to a single $z_\alpha$ only requires re-computing the marginals $\{\hat{\mu}_\gamma\}_{\gamma \supseteq \alpha}$ that depend on it. This lazy updating strategy avoids the cost of a full reconstruction at each step.

\subsection{Conditional ResidualPlanner}

This section addresses the problem of optimally allocating a privacy budget when measuring the residuals associated with a target marginal $\ma$. A baseline approach, which we term the ``iid" strategy, involves adding i.i.d. Gaussian noise to the marginal $\mu_\gamma$ and then deriving the noisy residuals from this single measurement \citep{grem}. A natural question is whether the fixed noise allocation implied by this strategy is optimal, or if a more targeted allocation could achieve lower error for the same privacy cost.

The potential for improvement is especially apparent in iterative or adaptive settings. In these scenarios, some residuals $\ra \subseteq \ma$ may have already been estimated in previous rounds. Intuitively, it is advantageous to allocate less of the current privacy budget to residuals that are already well-estimated, and more budget to those that are poorly estimated or unmeasured. To formalize this, we adapt machinery from ResidualPlanner \citep{resplanner}, which established the privacy cost and expected error for residual measurements.

\begin{proposition}\textup{\citep{resplanner}} \label{prop:crp_p}
Suppose noisy residuals $z_\ra = \zeta_\ra + \mathcal{N}(0, \sigma^2_\ra V_\ra)$ are released for each $\ra \subseteq \ma$. This satisfies $\rho$-zCDP for $\rho = \sum_{\ra \subseteq \ma} \frac{p_\ra}{2\sigma^2_\ra}$, where $p_\ra = \prod_{i \in\ra}\frac{n_i - 1}{n_i}$.
\end{proposition}

\begin{proposition}\textup{\citep{resplanner}} \label{prop:expected_error}
Suppose noisy residuals $z_\ra = \zeta_\ra + \mathcal{N}(0, \sigma^2 _\ra V_\ra)$ are measured for each $\ra \subseteq \ma$. The expected squared error $\mathbb{E}[\|\hat \mu_\ma - \mu_\ma\|_2^2]$ of the reconstruction $\hat \mu_\ma = \sum_{\ra \subseteq \ma} \texttt{Recon}(z_\ra, \ra, \ma)$ is $\sum_{\ra \subseteq \ma} v_\ra \sigma^2_\ra$, where $v_\ra = \left(\prod_{i \in \ra} \frac{n_i - 1}{n_i}\right)\left(\prod_{j \in \ma\setminus\ra}\frac{1}{n_j^2}\right)$.
\end{proposition}

Our goal is to determine the optimal noise variances $\{\sigma^2_\ra\}_{\ra \subseteq \ma}$ for a new round of measurements. We aim to minimize the final reconstruction error by combining our new measurements with prior information, subject to a privacy budget $\rho$ for the current round only. Let $\tilde{\sigma}^2_\ra$ be the variance of the existing estimate for residual $\ra$ from previous rounds (if no prior estimate exists, $\tilde{\sigma}^2_\ra = \infty$). A new measurement with variance $\sigma^2_\ra$ yields a combined estimate with an updated variance $\bar{\sigma}^2_\ra$, calculated by summing the precisions: $\bar{\sigma}^{-2}_\ra = \tilde{\sigma}^{-2}_\ra + \sigma^{-2}_\ra$. This leads to the Conditional ResidualPlanner (CRP) convex optimization problem:
\begin{equation} \label{eq:crp_problem}
\begin{split}
    \argmin_{\{\sigma^2_\ra > 0\}_{\ra \subseteq \ma}} \quad &\sum_{\ra \subseteq \ma} v_\ra \bar{\sigma}^2_\ra = \sum_{\ra \subseteq \ma} \frac{v_\ra}{\frac{1}{\sigma^2_\ra}+ \frac{1}{\tilde{\sigma}^2_\ra}} \\
    \text{s.t.} \quad &\sum_{\ra \subseteq \ma} \frac{p_\ra}{\sigma^2_\ra} \leq 2 \rho.
\end{split}
\end{equation}
The problem is convex when re-parameterized in terms of precisions ($x_\ra = 1/\sigma^2_\ra$) and can be solved efficiently as a second-order cone program \citep{boyd2004convex}.

Solving this optimization problem enables significant computational speedups. If the optimal variance $\sigma^2_\ra$ for a residual is very large, we can omit the measurement entirely, avoiding all the corresponding updates. This occurs frequently for low-order residuals that are already well-estimated from prior rounds like the total query. Furthermore, the optimization problem is solved with an efficient routine that first computes a closed-form analytical solution to an unconstrained version of the problem. In the common case that this solution is valid, a numerical solver is avoided completely. If not, the analytical result provides an excellent warm-start for an iterative solver (CVXPY with CLARABEL \citep{cvxpy, Clarabel_2024}), accelerating convergence. The full details of the cutoff for omitting measurements and the optimization strategy are in Appendix~\ref{app:crp_solver}.

\section{AIM+GReM}

To scale to high-dimensional datasets, we propose AIM+GReM: a mechanism that combines components from AIM+PGM\footnote{For clarity, we refer to the original AIM mechanism as AIM+PGM to highlight the reconstruction method.} that adapt the selection and measurement steps to the workload, privacy budget, and data domain with GReM-MLE \citep{grem} for efficient reconstruction. 
AIM+GReM is shown in Alg.~\ref{alg:aimGReM}. 
In Appendix~\ref{app:aim}, we detail components of AIM+PGM otherwise not described in this section.

The following components are shared between AIM+GReM and AIM+PGM.
Line 3 initializes the mechanism by measuring all 1-way marginals in $\W^\downarrow$ using the Gaussian mechanism with noise scale $\sigma_0^2$, which spends $d\rho_{init}$ of the privacy budget (see Appendix~\ref{app:aim_grem}).
Line 8 privately selects the next marginal to measure using the exponential mechanism with a score function that estimates the potential improvement in $L_1$ error for the final workload answers \citep{aim}.
Line 16 updates the privacy parameters for the next round using budget annealing, a heuristic that dynamically adjusts the budget based on whether the previous measurement yielded a significant improvement in accuracy (see Appendix~\ref{app:aim}).
In the remainder of this section, we describe the components on which AIM+GReM and AIM+PGM differ and consider how these differences affect utility, scalability, and privacy.  

\begin{algorithm*}[ht]
   \caption{AIM+GReM}
   \label{alg:aimGReM}
   \begin{algorithmic}[1]
    \Require Workload $\W$, Budget $\rho$, Data Marginals $\mu_\gamma$ for $\gamma \in \W^{\downarrow}$
    \Ensure Noisy marginals $\hat \mu_\ma$ for $\ma \in \W$
    \State Initialize $z_\tau = \mathbf{0}$, $\lambda_\tau = 0$ for $\tau \in \W^{\downarrow}$ \Comment{Initialize residual estimates and precisions}
    \State Set $\sigma_0^2 = \frac{| \W^\downarrow |}{0.9 \rho}$, $\rho_\init = \frac{1}{2 \sigma_0^2}$
    \State $z_\emptyset, \lambda_\emptyset, \ldots, z_{\{d\}}, \lambda_{\{d\}} \gets$ measure 1-way marginals with budget $\rho_\init$ each and convert to residuals
    \State Set $\rho_\used = d \rho_\init; \ t = d; \ \sigma_{t+1}^2 = \sigma_0^2;$ $\epsilon_{t+1} = \sqrt{\nicefrac{0.4 \rho}{\abs{\W^\downarrow}} }$
    \State Set $\hat \mu_{\ma} = \sum_{\ra \subseteq \ma} \Rec{z_\ra}{\ra}{\ma}$ for each $\ma \in \W$ \Comment{Initialize marginal estimates}
    \While{$\rho_\used < \rho$}
    \State $t = t + 1$
    \State Select $\ma_t \in \W^\downarrow$ using exponential mechanism with budget $\epsilon_t$ and weighted \Comment{Select}
    \Statex\hspace*{\algorithmicindent}expected improvement score \citep{aim}
    \State Set $\sigma_{\ra,t}^2$ for $\ra \subseteq \ma_t$ by solving CRP with budget $\frac{1}{2 \sigma_t^2}$ \Comment{Optimize}
    \For{$\ra \subseteq \ma_t$}
        \State $z_{\ra,t} = \measure(\mu_\ra, \sigma^2_{\ra,t})$ \Comment{Measure}
        \State $z'_\tau = \frac{\lambda_\tau z_\tau +  \sigma_{\tau,t}^{-2}z_{\tau,t}}{\lambda_\tau +  \sigma_{\tau,t}^{-2}}$
        \State $\hat \mu_\gamma = \hat \mu_\gamma + \Rec{z'_\tau - z_\tau}{\tau}{\gamma}$ for each $\ma \in \W$ such that $\ma \supseteq \ra$ \Comment{Lazy reconstruct}
        \State $z_\tau = z'_\tau$, $\lambda_\tau = \lambda_\tau + \sigma^{-2}_{\tau, t}$
        \EndFor
    \State Update $\rho_\used = \rho_\used + \frac{1}{2 \sigma_t^2} + \frac{\epsilon_t^2}{8}$
    \State Set $\sigma_{t + 1}^2$, $\epsilon_{t + 1}$ using budget annealing
    \EndWhile
\end{algorithmic}
\end{algorithm*}

\subsection{Differences with AIM+PGM} \label{s:diffs}

Despite sharing a similar structure, AIM+GReM and AIM+PGM differ on aspects of the select, measure, and reconstruct steps. 

\textbf{Tractable selection}. 
Since Private-PGM reconstruction can become intractable if the overall set of measured marginals causes the graphical model to become too complex, AIM+PGM filters out candidate marginals from $\W^\downarrow$ that increase the model size beyond a certain threshold.
In contrast, AIM+GReM does not filter the candidate set $\W^\downarrow$, since the complexity of reconstruction with GReM-MLE does not depend on which specific marginals are measured but, rather, on the size of the marginals in $\W$ to be reconstructed. 

\textbf{Measurement noise optimization}. 
When measuring a marginal, AIM+PGM uses the Gaussian Mechanism with isotropic noise. 
GReM was designed to take either marginals measured with isotropic noise or residuals measured with noise $\mathcal{N}(0, \sigma_\ra^2V_\ra)$. 
This allows us to optimize the measurement noise scaled on the residual level in AIM+GreM by solving a CRP problem in each iteration (Line 9).

\textbf{Reconstruction method}.
In lines 5 and 12-14 of Alg.~\ref{alg:aimGReM}, AIM+GReM uses GReM-MLE to reconstruct answers to marginals in $\W^\downarrow$ from prior measurements, while AIM+PGM uses Private-PGM.

\subsection{Utility, Scalability, and Privacy}

\textbf{Utility and scalability}.
For an adaptive query answering mechanism, Private-PGM reconstruction presents a utility-scalability tradeoff. For a fixed set of measurements, when tractable, Private-PGM has lower error than GReM-MLE; however, GReM-MLE is often tractable when Private-PGM is not \citep{grem}. Since AIM+PGM filters candidates for tractability, it only uses a subset of the marginals in $\W^\downarrow$, limiting its capacity to model the data and the utility of its answers. Because AIM+GReM does not filter candidates, it can better model the data and, in some cases, produce lower-error answers than AIM+PGM.

\textbf{Privacy}.
The privacy analysis for AIM+GReM follows as a corollary from the privacy analysis of AIM+PGM: Theorem 1 of \citet{aim}.
We provide a proof of Theorem~\ref{thm:privacy} in Appendix~\ref{app:aim_grem_priv}. 

\begin{restatable}{theorem}{privThm} \label{thm:privacy}
   For $\rho > 0$, \textup{AIM+GReM} satisfies $\rho$-zCDP.
\end{restatable}

\section{Experiments}
In this section, we empirically evaluate the performance of our proposed methods.

\textbf{Computing marginals and residuals}.
We compare the runtime of our in-axis operations against a baseline using traditional vector-based transformations via Kronecker matrix-vector multiplication~\citep{plateau1985stochastic, hdmm, resplanner}. We test two settings: (1) fixing attribute size $n_i = 16$ and varying marginal degree $|\gamma|$ from 2 to 7; and (2) fixing $|\gamma| = 3$ and varying $n_i$ from 2 to 256. For each setting, we benchmark decomposition ($\zeta_\gamma = \De{\mu_\gamma}{\gamma}{\gamma}$) and reconstruction ($\Rec{\zeta_\gamma}{\gamma}{\gamma}$).

\begin{figure*}
    \centering
    \includegraphics[width=1\textwidth]{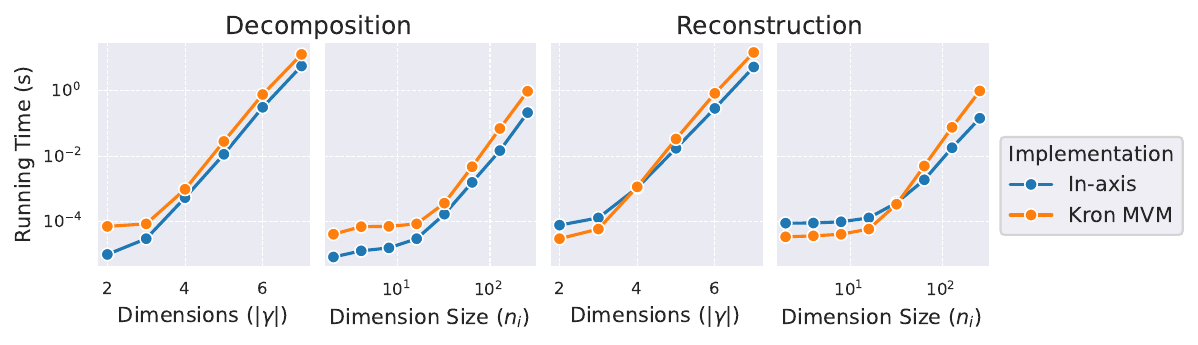}
    \caption{Running time in seconds of the $\Den$ and $\Recn$ operations with our in-axis implementation versus the baseline Kronecker-based approach.}
    \label{fig:speed}
\end{figure*}

As shown in Figure~\ref{fig:speed}, our in-axis decomposition is faster across all settings, achieving average speedups of 3.14x for fixed attribute size and 3.75x for a fixed marginal degree.
For reconstruction, our method is slightly slower than the baseline for low-degree marginals and small attribute sizes (max difference ~$5 \times 10^{-5}$s) but substantially faster for high-degree marginals and large attribute sizes (up to 9.75s faster).

\textbf{Algorithmic improvements}.
We conduct an ablation study to isolate the impact of our algorithmic improvements, comparing our proposed \texttt{lazy+CRP} against two baselines: \texttt{lazy+IID} (lazy updates with isotropic noise) and \texttt{full+IID} (full reconstruction with isotropic noise). To ensure a fair comparison, we use a fixed, pre-determined sequence of 3-way marginal measurements and a 50/50 privacy budget split (1-way, then 3-way marginals). We evaluate runtime and final average $L_1$ error on the Adult \citep{adult}, ACS \citep{acs}, and Loans \citep{hay2016principled} datasets (details in Appendix~\ref{app:exp_details}).

\begin{table*}[t]
  \centering
  \caption{Running time (s) and workload error of different strategies. Results are averaged over five trials with $\epsilon = 1$ and $\delta = 10^{-9}$.}
  \label{table:algo_improvements}
  \begin{tabular}{lccccc}
    \toprule
    & \multicolumn{3}{c}{Running time (s)} & \multicolumn{2}{c}{Workload error} \\
    \cmidrule(lr){2-4} \cmidrule(lr){5-6}
    Dataset & \texttt{lazy+CRP} & \texttt{lazy+IID} & \texttt{full+IID} & \texttt{CRP} & \texttt{IID} \\
    \midrule
    Adult & 8.10   & 16.58    & 149.64    & 9.532   & 10.166  \\
    ACS   & 311.60 & 607.97   & 14094.69  & 0.053   & 0.060   \\
    Loans & 37.87  & 489.46   & 3506.93   & 622.507 & 641.795 \\
    \bottomrule
  \end{tabular}
\end{table*}

Table \ref{table:algo_improvements} shows two distinct sources of performance gain. First, lazy updating is significantly faster than full reconstruction (7.2x–23.2x speedup). Second, when using lazy updates, our CRP strategy is substantially faster than IID measurement (2.0x–12.9x speedup). This is because CRP strategically allocates a negligible privacy budget to many low-order residuals, which allows us to omit those measurements and their costly updates entirely. Taken together, \texttt{lazy+CRP} achieves overall speedups of 18.5x (Adult), 45.2x (ACS), and 92.5x (Loans) over the \texttt{full+IID} baseline. On the utility front, CRP measurement also provides modest but consistent error improvements (1.03x–1.13x) over IID noise.

\textbf{Scalable query answering}.
We compare the utility-scalability tradeoff of AIM+GReM against AIM+PGM \citep{aim} and ResidualPlanner \citep{resplanner} on answering all 3-way marginals. We use a 50MB model for AIM+PGM, which provided the best utility-runtime trade-off in preliminary experiments (Appendix~\ref{app:additional_experiments}). To test scalability, we use the original datasets without the binning performed in the original AIM paper. We evaluate the mean $L_1$ error vs. total runtime across a range of privacy budgets ($\delta = 10^{-9}$). We use a larger budget range for the Loans dataset to demonstrate that ResidualPlanner and AIM+GReM do eventually approach or improve on the error of AIM+PGM as $\epsilon$ goes to infinity.

\begin{figure*}
    \centering
    \includegraphics[width=\textwidth]{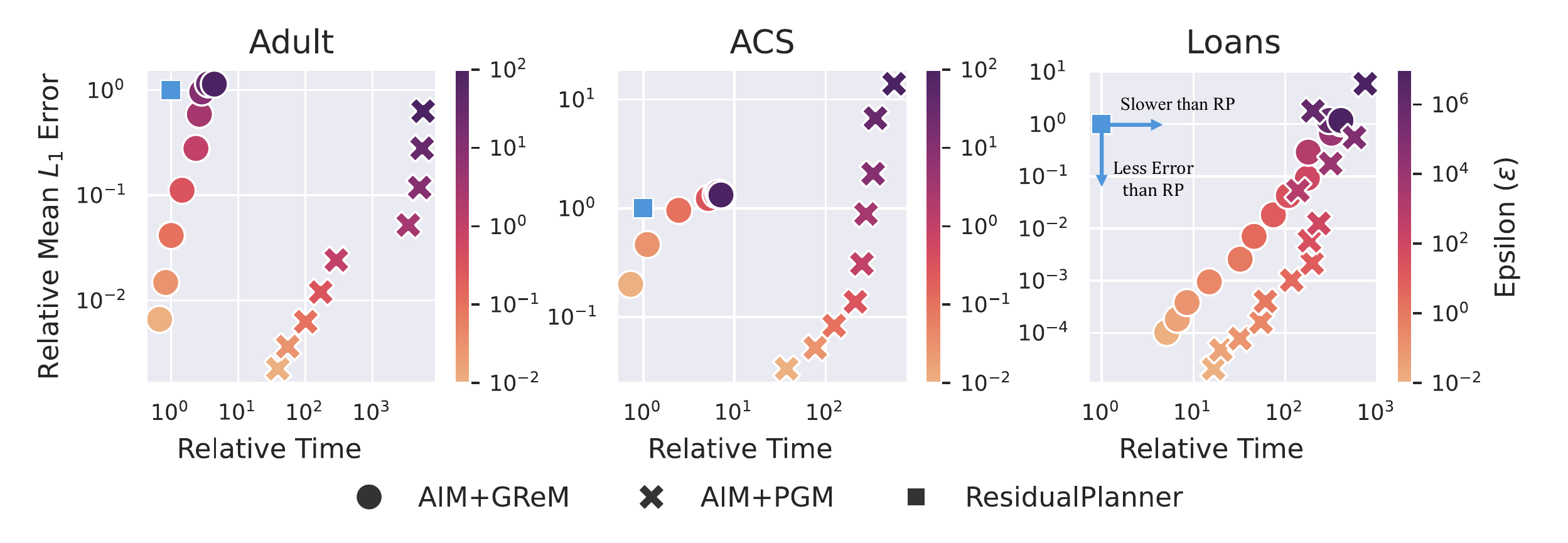}
    \caption{\small Mean $L_1$ error of AIM variants by running time on all 3-way marginals, relative to ResidualPlanner. Results are averaged across five trials. Privacy budgets $\epsilon$ are indicated in the vertical scale, with $\delta = 10^{-9}$. AIM+PGM uses a 50MB model. For 2D plots with absolute units see Appendix~\cref{app:additional_experiments}}
    \label{fig:aimGrem_l1_by_time}
\end{figure*}

\cref{fig:aimGrem_l1_by_time} visualizes the three-way tradeoff between privacy, utility, and runtime. To make this complex relationship easier to interpret, we present the time and error of the AIM variants relative to the ResidualPlanner baseline, with the privacy budget shown as a color gradient. For completeness, traditional 2D plots with absolute units are provided in Appendix~\cref{app:additional_experiments}. The figure shows that AIM+GReM improves the scalability-utility tradeoff. It is orders of magnitude faster than AIM+PGM (up to 697x on Adult) while achieving competitive error. This speedup is especially impactful for ACS where the running time at high epsilon (epsilon = 10) is reduced to about 5 hours with AIM+GReM from over a week with AIM+PGM. Compared to ResidualPlanner, AIM+GReM often has lower error (especially at low $\epsilon$) and is at worst only slightly higher (1.15x–1.37x). In contrast, AIM+PGM's error is significantly higher at high $\epsilon$ on the larger-domain datasets, up to 14.0x that of ResidualPlanner.

\section{Related Work}

Answering linear queries under differential privacy (DP) is a foundational problem, with early work like \citet{og_query_release} appearing soon after DP's introduction. Much of this work used strategies based on query decomposition, including hierarchical methods \citep{Hay2010hists, hierarchical_histograms}, wavelet-based techniques \citep{wavelets}, and the low-rank mechanism \citep{lowrank}. The matrix mechanism \citep{matrixmech} formalized this decomposition approach, a framing adopted by subsequent work like \citep{resplanner,factorization_mechanisms,privtree}.
Another line of research focuses on data-dependent mechanisms, an area with significant development since early works like the Multiplicative Weights Exponential Mechanism (MWEM) \citep{mwem}. We focus on the Adaptive Iterative Mechanism (AIM) as it performed best in recent benchmarks \citep{chen2025benchmarkingdifferentiallyprivatetabular}. Other notable contributions in this area include \citep{liu2021iterative,RAP,privsyn,privmrf}.

We also considered other reconstruction methods for comparison with GReM-MLE in this work. These included Approx-PGM \citep{approxpgm}, which attempts to speed up Private-PGM by relaxing consistency requirements. Another was GReM-LNN \citep{grem}, which tries to improve GReM's accuracy by imposing a local non-negativity constraint. However, preliminary experiments showed that GReM-LNN substantially increased GReM's runtime while not matching Private-PGM's accuracy. Approx-PGM did not meaningfully improve PGM's runtime in our setting and also reduced accuracy. Therefore, we did not include these methods in our experiments. In our experiments, we also considered GEM \citep{liu2021iterative}, but were unable to scale the implementation to our workloads, encountering GPU memory limitations.

\section{Discussion}
This research enhances differentially private adaptive query answering by improving the use of residual queries. We introduce an in-axis formulation for faster and simpler marginal/residual computations, a "lazy updating" strategy for GReM-based reconstruction that accelerates iterative updates by minimizing redundant calculations, and the conditional ResidualPlanner (CRP) for optimized, measurement-aware privacy budget allocation, reducing error and improving speed.
These advancements are integrated into AIM+GReM, a novel mechanism significantly improving scalability and efficiency while maintaining competitive accuracy. AIM+GReM is orders of magnitude faster than AIM+PGM, offering a vital alternative for large data domains where AIM+PGM struggles. While AIM+PGM may yield lower error in low-budget scenarios, AIM+GReM provides a favorable utility-scalability trade-off, generally outperforming ResidualPlanner at low privacy budgets and avoiding AIM+PGM's high error on larger datasets at high budgets.

\section*{Acknowledgments}
This work utilized resources from Unity, a collaborative, multi-institutional high-performance computing cluster managed by UMass Amherst Research Computing and Data. 
This material is based upon work supported by the National Science Foundation under Grant Nos. 2046235, 2210979 and 2317232. 

\FloatBarrier

\bibliographystyle{ACM-Reference-Format}
\bibliography{dp}

\clearpage
\appendix
\thispagestyle{empty}

\onecolumn
\aistatstitle{Supplementary}

\section{Marginals and Residuals Example} \label{app:example}

Let us continue the example from Section \ref{s:preliminaries}.
For a demographics dataset, if age group (Age) takes on values in \{18-24, 25-34, 35-50, 50+\} and education level (Educ) takes on values \{no college, some college, obtained degree\}, then the marginal $\mu_{\{\text{Age}, \text{Educ}\}}$ is a $4 \times 3$ array. 
Suppose we observe $\mu_{\{\text{Age}, \text{Educ}\}}$:
\begin{equation*}
    \mu_{\{\text{Age}, \text{Educ}\}} = 
    \begin{bmatrix}
        7   & 5  & 2  \\
        3   & 5  & 11 \\
        10  & 2  & 11 \\
        9   & 18 & 17
    \end{bmatrix}
\end{equation*}

Let us compute the residuals $\zeta_\ra$ for $\ra \subseteq \{\text{Age}, \text{Educ}\}$ using the Decomp function.

\begin{equation*}
    \zeta_\emptyset = \text{Decomp}(\mu_{\{\text{Age}, \text{Educ}\}}, {\{\text{Age}, \text{Educ}\}}, \emptyset) = 
    \begin{bmatrix}
        100
    \end{bmatrix}
\end{equation*}

\begin{equation*}
    \zeta_{\{\text{Age}\}} = \text{Decomp}(\mu_{\{\text{Age}, \text{Educ}\}}, {\{\text{Age}, \text{Educ}\}}, \{\text{Age}\}) = 
    \begin{bmatrix}
        5 \\
        9 \\
        30
    \end{bmatrix}
\end{equation*}

\begin{equation*}
    \zeta_{\{\text{Educ}\}} = \text{Decomp}(\mu_{\{\text{Age}, \text{Educ}\}}, {\{\text{Age}, \text{Educ}\}}, \{\text{Educ}\}) = 
    \begin{bmatrix}
        1 & 12
    \end{bmatrix}
\end{equation*}

\begin{equation*}
    \zeta_{\{\text{Age}, \text{Educ}\}} = \text{Decomp}(\mu_{\{\text{Age}, \text{Educ}\}}, {\{\text{Age}, \text{Educ}\}}, \{\text{Age}, \text{Educ}\}) = 
    \begin{bmatrix}
        4   &  13 \\
        -6  &   6 \\
        11  &  13
    \end{bmatrix}
\end{equation*}

To reconstruct the marginal $\mu_{\{\text{Age}, \text{Educ}\}}$, we can map each of the residuals $\zeta_\ra$ above to a marginal component $\mu_{\{\text{Age}, \text{Educ}\}, \ra}$ for $\ra \subseteq \{\text{Age}, \text{Educ}\}$ using the Recon function. Recall that a marginal is the sum of its components: $\mu_{\{\text{Age}, \text{Educ}\}} = \sum_{\ra \subseteq \{\text{Age}, \text{Educ}\}} \mu_{\{\text{Age}, \text{Educ}\}, \ra} $.

\begin{align*}
    \mu_{\{\text{Age}, \text{Educ}\}, \emptyset} &= \text{Recon}(\zeta_\emptyset, \emptyset, \{\text{Age}, \text{Educ}\})  \\
    &= \frac{1}{12}
    \begin{bmatrix}
        100 & 100 & 100 \\
        100 & 100 & 100 \\
        100 & 100 & 100 \\
        100 & 100 & 100 
    \end{bmatrix}
\end{align*}

\begin{align*}
    \mu_{\{\text{Age}, \text{Educ}\}, \{\text{Age}\}} &= \text{Recon}(\zeta_{\{\text{Age}\}}, \{\text{Age}\}, \{\text{Age}, \text{Educ}\})  \\
    &= \frac{1}{12}
    \begin{bmatrix}
        -44 & -44 & -44 \\
        -24 & -24 & -24 \\
        -8  & -8  & -8  \\
        76  & 76  & 76 
    \end{bmatrix}
\end{align*}

\begin{align*}
    \mu_{\{\text{Age}, \text{Educ}\}, \{\text{Educ}\}} &= \text{Recon}(\zeta_{\{\text{Educ}\}}, \{\text{Educ}\}, \{\text{Age}, \text{Educ}\})  \\
    &= \frac{1}{12}
    \begin{bmatrix}
        -13 & -10 & 23 \\
        -13 & -10 & 23 \\
        -13 & -10 & 23 \\
        -13 & -10 & 23 
    \end{bmatrix}
\end{align*}

\begin{align*}
    \mu_{\{\text{Age}, \text{Educ}\}, \{\text{Age}, \text{Educ}\}} &= \text{Recon}(\zeta_{\{\text{Age}, \text{Educ}\}}, \{\text{Age}, \text{Educ}\}, \{\text{Age}, \text{Educ}\})  \\
    &= \frac{1}{12}
    \begin{bmatrix}
        41  & 14  & -55 \\
        -27 & -6  & 32  \\
        41  & -58 & 17  \\
        -55 & 50  & 5 
    \end{bmatrix}
\end{align*}

\section{Differential Privacy} \label{app:dp}

Differential privacy is a mathematical criterion of privacy that bounds the effect of any individual in the dataset on the output of a mechanism by adding noise to the computation.
We say that two datasets $\D, \D'$ are neighboring if they differ by at most one record.
This is referred to as the add/remove or unbounded neighborhood relation.

\begin{definition}(Differential Privacy; \cite{dwork2006calibrating,zcdp})
  Let $\M: \X \rightarrow \Y$ be a randomized mechanism. For any neighboring datasets $\D, \D'$ that differ by at most one record, denoted $\D \sim \D'$, and all measurable subsets $S \subseteq \Y$:
  \begin{itemize}
    \item if $\Pr(\M(\D) \in S) \leq \exp(\epsilon) \cdot \Pr(\M(\D') \in S) + \delta$, then $\M$ satisfies $(\epsilon, \delta)$-approximate differential privacy, denoted $(\epsilon, \delta)$-DP;
    \item if $D_\ma(\M(\D) || \M(\D')) \leq \rho \ma$ for all $\ma \in (1, \infty)$ where $D_\ma$ is the $\ma$-Renyi divergence between distributions $\M(\D), \M(\D')$, then $\M$ satisfies $\rho$-zCDP.
  \end{itemize}
\end{definition}

While $(\epsilon, \delta)$-DP is a more common notion, it is often more convenient to work with zCDP. 
There exists a conversion from zCDP to $(\epsilon, \delta)$-DP. 

\begin{proposition}[zCDP to DP Conversion; \cite{canonne2020discrete}] \label{prop:zcdp_conversion}
  If mechanism $\M$ satisfies $\rho$-zCDP, then $\M$ satisfies $(\epsilon, \delta)$-DP for any $\epsilon > 0$ and $\delta = \min_{\alpha > 1} \frac{\exp((\alpha - 1)(\alpha\rho-\epsilon))}{\alpha - 1}\left(1 - \frac{1}{\alpha}\right)^\alpha$.
\end{proposition}

Next, we introduce two building block mechanisms. 
An important quantity in analyzing the privacy of a mechanism is sensitivity. 
The $L_k$ sensitivity of a function $f: \X \rightarrow \R$ is given by $\Delta_k(f) = \max_{\D \sim \D'} \norm{f(\D) - f(\D')}_k$. 
If $f$ is clear from the context, we write $\Delta_k$. 
For a marginal $\mu_\gamma$ with any $\gamma \subseteq [d]$, $\Delta_k(\mu_\gamma) = 1$.

\begin{proposition}[zCDP of exponential mechanism; \cite{mcsherry2007mechanism,cesar2021bounding}]\label{defn:exponential mech}
  Let $\epsilon > 0$ and $\text{Score}_r: \X \rightarrow \R$ be a quality score of candidate $r \in \mathcal{R}$ for dataset $\D$. Then the exponential mechanism outputs a candidate $r \in \mathcal{R}$ according to the following distribution:
  $\Pr(\M(\D) = r) \propto \exp \big(\frac{\epsilon}{2 \Delta_1(f)} \text{Score}_r(\D) \big)$.
  The exponential mechanism satisfies $\frac{\epsilon^2}{8}$-zCDP.
\end{proposition}

\begin{proposition}[zCDP of Gaussian mechanism; \cite{resplanner}]\label{defn:gaussian mech}
  Let $f$ be a linear function of dataset $\D$, represented by matrix $M$. Given dataset $\D$, the Gaussian mechanism adds correlated Gaussian noise to $f(\D)$ with covariance $\Sigma$, i.e., $\M(\D) = f(\D) + \N(0, \Sigma)$. Then the Gaussian Mechanism satisfies $\frac{\eta}{2}$-zCDP, where the privacy cost $\eta$ is the largest diagonal element of the matrix $M^\top \Sigma^{-1} M$.
\end{proposition}

Adaptive composition and postprocessing are two important properties of differential privacy that allow us to construct complex mechanisms from the above building blocks.

\begin{proposition}[zCDP Properties; \cite{zcdp, whitehouse2023fully}] \label{prop:zcdpProperties}
  zCDP satisfies two properties of differential privacy:
  \begin{enumerate}
    \item (Fully Adaptive Composition) 
    Let $\M_1: \X \rightarrow \mathcal{Y}_1$ satisfy $\rho_1$-zCDP and $\M_2: \X \times \mathcal{Y}_1 \rightarrow \mathcal{Y}_2$ satisfy $\rho_2$-zCDP. Then $\M = \M_2(\D, \M_1(\D))$ satisfies $(\rho_1 + \rho_2)$-zCDP.

    \item (Postprocessing) Let $\M_1: \X \rightarrow \mathcal{Y}$ satisfy $\rho$-zCDP and $f: \mathcal{Y} \rightarrow \mathcal{Z}$ be a randomized algorithm. Then $\M: \X \rightarrow \mathcal{Z} = f \circ \M_1$ satisfies $\rho$-zCDP.
  \end{enumerate}
\end{proposition}

\section{Residual Properties} \label{app:res_props}

In this section we provide background on residuals and marginals and demonstrate that our multi-dimensional array perspective in \cref{sec:residuals-and-marginals} is equivalent to the original definitions.
Residuals and marginals are usually defined using Kronecker product matrices.
We will first review a well known correspondence between Kronecker product matrix-vector multiplications and $n$-mode products of tensors. 
We will then see that the $n$-mode products that arise with marginals and residuals correspond to the simple and interpretable multi-dimensional array operations from \cref{sec:residuals-and-marginals}.

\subsection{Background on Kronecker product matrices and tensor \texorpdfstring{$n$-mode}{n-mode} products}

We first establish the necessary background on Kronecker product matrices and tensor $n$-mode products.

\textbf{Kronecker products}.
If a matrix $A$ can be written as
$$
A = A_1 \otimes A_2 \otimes \cdots \otimes A_d = \bigotimes_{i=1}^d A_i
$$
we will call it a Kronecker-product matrix with factors $A_1, \ldots, A_d$. In this paper we will only consider Kronecker-product matrices with $d$ factors where $d$ is the number of attributes of the dataset. 

We will use the following fact: if $A = \bigotimes_{i=1}^d A_i$ and $B = \bigotimes_{i=1}^d B_i$ and the sizes of each $A_i$ and $B_i$ are compatible, then $AB = \bigotimes_{i=1}^d (A_i B_i)$.

\textbf{Tensors and $n$-mode products}.
We now introduce tensors and $n$-mode products. A comprehensive background can be found in~\cite{kolda2009tensor}.
Let $\mathcal U \in \mathbb R^{n_1 \times \cdots \times n_d}$ be a tensor or multi-dimensional array. We will use these two terms interchangeably. 
A mode of a tensor is one of the dimensions. In an array language like NumPy, this is called an axis.

The mode-$k$ fibers of a tensor are the vectors along the $k$th mode. For each fixed tuple $i_1, \ldots, i_{k-1}, i_{k+1}, \ldots, i_d$ of indices for other dimensions, the values $u_{i_1 \cdots i_{k-1}\  j\  i_{k+1} \cdots i_d}$ for $j$ ranging from $1$ to $n_k$, taken as a column vector, form one mode-$k$ fiber. There are $\prod_{i \neq k} n_i$ mode-$k$ fibers in total.

An $n$-mode product is a linear transformation of the tensor along one of its modes, where $n$ indicates the mode (e.g., a 1-mode product, 2-mode product, etc.).
We will say  ``$n$-mode product'' for the generic operation but ``$k$-mode product'' for the operation applied to a specific mode $k$, because we have already used $n$ for another purpose in this paper.

The \emph{$k$-mode product} $\mathcal{U} \times_k A$ with matrix $A \in \R^{m \times n_k}$ applies the linear transformation $A$ to every mode-$k$ fiber of $\mathcal U$. The result is a tensor $\mathcal V \in \mathbb R^{n_1 \times \cdots n_{k-1} \times m \times n_{k+1} \times \cdots n_d}$ with entries 
$$
v_{\cdots \,i\,\cdots} = \sum_{j=1}^{n_j}a_{ij} \cdot u_{\cdots \, j \, \cdots}
$$
where each ellipsis suppresses shared indices between the LHS and RHS and the indices $i$ and $j$ appear in the $k$th position. So, $v_{\cdots \, i \, \cdots} = v_{i_1 \cdots i_{k-1}\ i\  i_{k+1} \cdots i_d}$ and $u_{\cdots \, j \, \cdots} = u_{i_1 \cdots i_{k-1}\ j\  i_{k+1} \cdots i_d}$ for fixed $i_1, \ldots, i_{k-1}, i_{k+1}, \ldots, i_d$.

A sequence of $n$-mode products is written as
$$
\mathcal U \times_{i=1}^d A_i  \ := \mathcal U \times_1 A_1 \times_2 \cdots \times_d A_d.
$$
Sequences of $n$-mode products can be performed in any order and give the same result.

For a tensor $\mathcal U \in \R^{n_1 \times \cdots \times n_d}$, the vectorization operation $\vec(\mathcal U)$ reshapes $\mathcal U$ into a vector $u$ of size $\prod_{i=1}^d n_i$. We will also define a tensorization operation $\ten(u)$ that reshapes $u$ back into a tensor of size $n_1 \times \cdots \times n_d$, with the sizes $n_1,\ldots, n_d$ implicit from context. This gives, for example, that $\mathcal U = \ten(\vec(\mathcal U))$.

\textbf{Kronecker product matrix-vector multiplication as sequence of $n$-mode products}.
Let $A = A_1 \otimes \cdots \otimes A_d$ be a Kronecker-product matrix with $A_i \in \R^{m_i \times n_i}$ and let $\mathcal U \in \R^{n_1 \times \cdots \times n_d}$ be a tensor of compatible size. The matrix-vector multiplication $A \,\vec(\mathcal U)$ is equivalent to a sequence of $n$-mode products \cite{kolda2009tensor}:
$$
A\,\vec(\mathcal U) = \vec(\mathcal U \times_{i=1}^d A_i).
$$

Another way to write this is:
$$
A u = \vec(\ten(u) \times_{i=1}^d A_i).
$$
for any vector $u \in \R^n$ where $n = \prod_{i=1}^d n_i$.

Either way, we see that the multiplication of a Kronecker-product matrix and a vector can be conceptualized as applying a sequence of $n$-mode products to the tensor corresponding to the vector.

\subsection{Marginals and Residuals}

In~\citet{resplanner} and \citet{grem}, marginals and residuals are defined as Kronecker-product structured queries on the data vector $\mathbf p = \vec(\mathcal A)$, i.e., the vectorized data array, which has length $n = \prod_{i=1}^d n_i$. 
A (vector-valued) query is a linear transformation $Q \in \R^{m \times n}$ that operates on $\p$ and gives query answers $Q \p \in \R^m$.

\textbf{Marginal queries and the \texttt{sum} operation}.
A marginal query $M_\gamma$ over attribute subset $\gamma \subseteq [d]$ is defined as:
$$
M_\gamma = \bigotimes_{k=1}^d 
			\begin{cases}
			\mathbb I_{n_k} & k \in \gamma \\
			1_{n_k}^\top & k \notin \gamma
			\end{cases},
$$
where $\mathbb I_{\ell}$ is the $\ell \times \ell$ identity matrix and $1_{\ell}$ is a column vector of $\ell$ ones. 
The marginal $\mu_\gamma = M_\gamma \p$ (shaped as a vector) is the vector of answers to the query $M_\gamma$.
If we interpret $M_\gamma \p$ through the lens of $n$-mode products, we see it is equivalent to applying a sequence of transformations to the data array $\mathcal A$: 
\begin{enumerate}[itemsep=0pt]
\item For each $k \in \gamma$, apply the identity transformation to dimension $k$,
\item For each $k \notin \gamma$, sum over dimension $k$. 
\end{enumerate}
This is exactly the definition of a marginal.

The operator $1_{n_k}^\top$, which sums over dimension $k$ of the data array, is the first of the four elementary operations used for marginal and residual operations, and corresponds to the \texttt{\_sum} operation in \cref{tab:axis_ops}.

An important technical comment is that the sequence of $n$-mode products to compute a marginal from the data array does \emph{not} reduce the number of dimensions of the multi-dimensional array. 
The sum operation $1_{n_k}^\top$ maps each mode-$k$ fiber to a scalar, but this dimension is preserved as a dimension of size $1$, or \emph{singleton dimension}, in the resulting array.
This allows us to follow the convention that every array---the data array $\mathcal A$ as well as each marginal $\mu_\gamma$ and, later, each residual array---has one dimension per attribute occuring in order from $1$ to $d$, with singleton dimensions for marginalized attributes.\footnote{In \cref{sec:residuals-and-marginals} we defined $\mu_\gamma$ as a $|\gamma|$-dimensional array, but we will now follow the convention that it is a $d$-dimensional array with singleton dimensions for attributes not in $\gamma$.}

\textbf{Residual queries and the \texttt{sub} operation}.
A residual query $R_\tau$ over attribute subset $\tau \subseteq [d]$ is defined as:
$$
R_\tau  = \bigotimes_{k=1}^d 
			\begin{cases}
			S_{(n_k)}  & k \in \tau \\
			1_{n_k}^\top & k \notin \tau
			\end{cases}
$$
where $S_{(\ell)} \in \R^{\ell-1 \times \ell}$ is a \emph{basis matrix} defined below.
The residual $\zeta_\tau = R_\tau \p$ (shaped as a vector) is the vector of answers to the residual query $R_\tau$.
There are different choices for which basis matrix to use. 
We use the \emph{subtraction matrix} $S_{(\ell)}$~\citep{resplanner}, defined as
$$
S_{(\ell)} = \begin{bmatrix}
-1 & 1 & 0 & \cdots & 0 & 0 \\
-1 & 0 & 1 & \cdots & 0 & 0 \\
\vdots & \vdots & \vdots & \ddots & \vdots & \vdots \\
-1 & 0 & 0 & \cdots &  1 & 0 \\
-1 & 0 & 0 & \cdots &  0 & 1
\end{bmatrix}
= \begin{bmatrix}
-1_{\ell-1} & \I_{\ell-1} \\
\end{bmatrix} \in \R^{\ell-1 \times \ell}.
$$
A different basis was used in~\cite{grem}. Our definition of the subtraction matrix has the opposite sign of the one in~\cite{resplanner}.\footnote{This makes the corresponding \texttt{\_sub} operation and its pseudoinverse \texttt{\_center} slightly more interpretable.}
With these definitions, we can see that the residual query operation $R_\tau \p$ is also equivalent to applying a sequence of transformations to the data array $\mathcal A$:
\begin{enumerate}[itemsep=0pt]
\item For each $k \in \tau$, apply the subtraction operator $S_{(n_k)}$ to dimension $k$,
\item For each $k \notin \tau$, sum over dimension $k$.
\end{enumerate}
We can also see from its definition that $S_{(\ell)}$ has a simple interpretation as a linear transformation. For any $v \in \R^\ell$,
$$
S_{(\ell)} v = v_{2:\ell} - v_{1}.
$$
This is implemented by the $\texttt{\_sub(v)}$ operation from~\cref{tab:axis_ops}.

\subsection{Converting between marginals and residuals: \texorpdfstring{$\Den$ and $\Recn$}{Den and Recn}}

Residual and marginal queries are closely related.
Both sum over unwanted dimensions of the data array, and marginal queries preserve the remaining dimensions by applying the identity operator, while residual queries apply the subtraction operator.

\textbf{Marginals to residuals with $\Den$}.
From this it is clear we can compute the $\tau$-residual from the $\gamma$-marginal for any $\tau \subseteq \gamma$ by summing over dimensions $k \in \gamma \setminus \tau$ and applying the subtraction operator to dimensions $k \in \tau$. 
Mathematically, for $\tau \subseteq \gamma$ we have that $R_\tau = T_{\tau,\gamma}M_\gamma$ where\footnote{
We can verify this is  using Kronecker algebra:
$$
T_{\tau,\gamma} M_\gamma \quad = \quad \bigotimes_{k=1}^d 
					\begin{cases}
					S_{(n_k)} \I_{n_k} & k \in \tau \\
					1_{n_k}^\top \I_{n_k} & k \in \gamma \setminus \tau \\
					1 \cdot 1_{n_k}^\top & k \notin \gamma
					\end{cases}
					\quad =  \quad
					\bigotimes_{k=1}^d
					\begin{cases}
					S_{(n_k)} & k \in \tau \\
					1_{n_k}^\top & k \notin \gamma
					\end{cases}
					\quad = \quad 
					R_\tau
$$}
\begin{equation}
  \label{eq:decomp}
  T_{\tau,\gamma} := \bigotimes_{k=1}^d 
					\begin{cases}
					S_{(n_k)} & k \in \tau \\
					1_{n_k}^\top & k \in \gamma \setminus \tau \\
					1 & k \notin \gamma
					\end{cases}.
\end{equation}
The operator $T_{\tau,\gamma}$ is therefore a linear transformation that converts a $\gamma$-marginal to a $\tau$-residual, since $\zeta_\tau = R_\tau \p = T_{\tau,\gamma} M_\gamma \p = T_{\tau,\gamma} \mu_\gamma$.

The transformation $T_{\tau,\gamma} \mu_\gamma$ is exactly the $\De{\mu_\gamma}{\gamma}{\tau}$ operation in~\cref{alg:decomp}.
Using the $n$-mode product perspective, we see from~\cref{eq:decomp} that $T_{\tau,\gamma} \mu_\gamma$ is equivalent to a sequence of operations on the marginal $\mu_\gamma$ (shaped as a $d$-dimensional array): 
\begin{enumerate}[itemsep=0pt]
\item For each $k \in \tau$, apply the subtraction operator $S_{(n_k)}$ to dimension $k$,
\item For each $k \in \gamma \setminus \tau$, sum over dimension $k$,
\item For each $k \notin \gamma$, preserve the singleton dimension $k$ by applying the scalar identity operator.
\end{enumerate}
This is exactly the definition of $\De{\mu_\gamma}{\gamma}{\tau}$.

\textbf{Residuals to marginals with $\Recn$}.
The $\Recn$ operation is the pseudoinverse of $\Den$:
\begin{equation}
\label{eq:recon}
T_{\tau,\gamma}^+ = \bigotimes_{k=1}^d
						\begin{cases}
						S_{(n_k)}^+ & k \in \tau \\
						\frac{1}{n_k} 1_{n_k} & k \in \gamma \setminus \tau \\
						1 & k \notin \gamma
						\end{cases}
\end{equation}
We used the rule that $(A_1 \otimes \cdots \otimes A_d)^+ = A_1^+ \otimes \cdots \otimes A_d^+$ and simple calculations for the pseudoinverses of $1_{n_k}^\top \in \R^{1 \times n_k}$ and the scalar $1$.
To map \cref{eq:recon} to the $\Recn$ operation in \cref{alg:recon}, it remains to characterize $S_{(n_k)}^+$ and $\frac{1}{n_k} 1_{n_k}$, which are the pseudoinverses of the subtraction and sum operations, as interpretable linear operators.

\emph{Center operator as pseudoinverse of subtraction operator.}
The pseudoinverse $S_{(\ell)}^+ \in \R^{\ell \times \ell-1}$ of the subtraction matrix is~\citep{resplanner}:
$$
S_{(\ell)}^+ = 
\begin{bmatrix}
0 & 0 & \cdots & 0 \\
1 & 0 & \cdots & 0 \\
0 & 1 & \cdots & 0 \\
\vdots & \vdots & \ddots & \vdots \\
0 & 0 & \cdots & 1 \\
\end{bmatrix}
-
\frac{1}{\ell} \begin{bmatrix}
1 & 1 & \cdots & 1 \\
1 & 1 & \cdots & 1 \\
1 & 1 & \cdots & 1 \\
\vdots & \vdots & \ddots & \vdots \\
1 & 1 & \cdots & 1 \\
\end{bmatrix}
=
\begin{bmatrix}
0_{\ell-1}^\top \\
\I_{\ell-1}
\end{bmatrix}
-
\frac{1}{\ell} 
1_{\ell} 1_{\ell-1}^\top 
$$	
Its operation on a vector $v \in \R^{\ell-1}$ is:
$$
S_{(\ell)}^+ v = 
\begin{bmatrix}
0 \\
v 
\end{bmatrix}
-
1_{\ell} \mu
$$
where $\mu = \frac{1}{\ell} \sum_{i=1}^{\ell-1} v_i$, which is also the mean of the zero-padded vector $[0, v]^\top$. 
Thus, this operation can be implemented in two steps as described in the $\texttt{\_center}$ operation in \cref{tab:axis_ops}: first pad $v$ with an initial zero entry, then subtract off the mean.

\emph{Smear operator as pseudoinverse of sum operator.}
The pseudoinverse of the sum operator $1_{n_k}^\top \in \R^{1 \times n_k}$ is the ``smear'' operator $\frac{1}{n_k} 1_{n_k} \in \R^{n_k \times 1}$. It operates on a scalar $v$ as
$$
v \mapsto \frac{1}{n_k} 
\begin{bmatrix}
v \\
\vdots \\
v \\	
\end{bmatrix}
$$
In words, it ``smears'' the value $v$ evenly across the $n_k$ entries of the resulting vector. 
This is implemented by the $\texttt{\_smear}$ operation in \cref{tab:axis_ops}.

\emph{Put together: $T_{\tau,\gamma}^+$ as the $\Recn$ operation.}
We can now interpret the operation $T_{\tau,\gamma}^+ \zeta_\tau$ as the following sequence of operations on a residual $\zeta_\tau$ (shaped as a $d$-dimensional array): 
\begin{enumerate}[itemsep=0pt]
\item For each $k \in \tau$, apply the center operator $S_{(n_k)}^+$ to dimension $k$,
\item For each $k \in \gamma \setminus \tau$, apply the smear operator $\frac{1}{n_k} 1_{n_k}$ to dimension $k$,
\item For each $k \notin \gamma$, preserve the singleton dimension $k$ by applying the scalar identity operator.
\end{enumerate}
This is exactly the definition of $\Rec{\zeta_\tau}{\tau}{\gamma}$ in \cref{alg:recon}.

\subsection{The invertible transformation \texorpdfstring{$T_\gamma$}{T\_gamma}}

Prior work~\citep{resplanner,grem} established the invertible linear transformation between the marginal $\mu_\gamma$ and the residuals $[\zeta_\tau]_{\tau \subseteq \gamma}$ stated in \cref{prop:bijection}.
We briefly elaborate using the notation defined in this section.
Define
$$
T_\gamma = 
\left[
\begin{array}{ccc}
& T_{\emptyset,\gamma}& \\[2pt]
\hline
&\vdots & \\
\hline
& T_{\gamma,\gamma} &  \\[2pt]
\end{array}
\right] \in \R^{n_\gamma \times n_\gamma}
$$
to be the vertical concatenation of the decomposition operators $T_{\tau,\gamma}$ for all $\tau \subseteq \gamma$, where we recall that $n_\gamma = \prod_{k \in \gamma} n_k$.
This matrix is invertible with inverse (e.g., see Theorem~2 of \citep{grem}):
$$
T_\gamma^{-1} = 
\left[
\begin{array}{c|c|c}
&&  \\
T_{\emptyset,\gamma}^+ & \cdots & T_{\gamma,\gamma}^+ \\
&&
\end{array}
\right].
$$
The operator $T_{\gamma}$ applied to a marginal $\mu_\gamma$ gives the $\tau$-residuals for all $\tau \subseteq \gamma$:
$$
T_\gamma \mu_\gamma =
\begin{bmatrix}
T_{\emptyset,\gamma} \mu_\gamma \\
\vdots \\
T_{\gamma, \gamma} \mu_\gamma 
\end{bmatrix}
= [T_{\tau,\gamma} \mu_\gamma]_{\tau \subseteq \gamma}
= [\De{\mu_\gamma}{\gamma}{\tau}]_{\tau \subseteq \gamma}
= [\zeta_\tau]_{\tau \subseteq \gamma}.
$$
The operator $T_\gamma^{-1}$ applied to the residuals $[\zeta_\tau]_{\tau \subseteq \gamma}$ gives the $\gamma$-marginal:
$$
T_\gamma^{-1}
\begin{bmatrix}
\zeta_\emptyset \\
\vdots \\
\zeta_\gamma
\end{bmatrix}
=
\sum_{\tau \subseteq \gamma} T_{\tau,\gamma}^+ \zeta_\tau
= \sum_{\tau \subseteq \gamma} \Rec{\zeta_\tau}{\tau}{\gamma}
= \mu_\gamma.
$$

\subsection{Residual Measurement}
With the equivalence between the decomposition function $\De{\cdot}{\gamma}{\tau}$ and the residual transformation matrix $T_{\tau,\gamma}$ established, we can implement residual measurement as detailed in \cref{alg:measure_res}.

\begin{algorithm}[ht]
    \caption{\measure$(\mu_\tau, \sigma_\tau^2)$~\citep{resplanner}}
    \label{alg:measure_res}
    \begin{algorithmic}[1]
        \Require Marginal $\mu_\ra$, variance $\sigma_\ra^2$
        \Ensure Noisy residual $z_\ra = \zeta_\ra + \mathcal N(0, \sigma_\ra^2 V_\ra)$
        \State $y_\ra = \mu_\ra + N(0, \sigma_\ra^2)$
        \State Return $\De{y_\ra}{\ra}{\ra}$
    \end{algorithmic}
\end{algorithm}

\section{Conditional ResidualPlanner} \label{app:crp}
This section provides a more detailed exposition of the Conditional ResidualPlanner (CRP) problem. We start by showing how our \cref{prop:crp_p} (regarding privacy guarantees) and \cref{prop:expected_error} (regarding error formulation) derive from results presented in \citet{resplanner}. Next, we detail the specifics of our algorithm for solving the CRP problem.

\subsection{CRP Parameters: Privacy and Error} \label{app:crp_parameters}
To establish the connection between our framework and the results in \citet{resplanner}, we first refer to \cref{app:res_props}. The correspondence between in-axis operations and Kronecker products demonstrates that our residual measurement algorithm, \cref{alg:measure_res}, is equivalent to the residual measuring mechanism $\mathcal{M}_A\paren{\mathbf{x}; \sigma_A^2}$ described in \citet{resplanner}.
With this equivalence established, we can leverage the theorems from \citet{resplanner}:

\textbf{Privacy Cost Derivation (\cref{prop:crp_p})}:
Theorem 4.5 in \cite{resplanner} states that the ``privacy cost'' (as defined therein) of employing a mechanism equivalent to our \cref{alg:measure_res} to measure a specific residual, denoted $\ra$, with noise variance $\sigma_\ra^2$, is $\frac{p_\ra}{\sigma^2_\ra}$.
Consequently, under this framework, our \cref{alg:measure_res} when applied to residual $\ra$ with noise variance $\sigma_\ra^2$ satisfies $\rho$-zero Concentrated Differential Privacy ($\rho$-zCDP) with:
\begin{equation} \label{eq:rho_derivation}
\rho = \frac{p_\ra}{2\sigma^2_\ra}
\end{equation}
This forms the basis for our \cref{prop:crp_p}.

\textbf{Error Formulation Derivation (\cref{prop:expected_error})}:
Similarly, Theorem 4.7 in \citet{resplanner} addresses the variance of marginal estimates. Let $\hat{\mu}_\ma$ be an estimate of a marginal $\ma$, constructed as $\hat{\mu}_\ma = \sum_{\ra \subseteq \ma} \Rectg{\hat{z}_\ra}$.

Theorem 4.7 states that the variance of each individual entry in the vector $\hat{\mu}_\ma$ is given by $\sum_{\ra \subseteq \ma} \sigma_\ra^2 v_\ra$.

Let $n_\ma$ be the dimensionality of the marginal $\ma$ (i.e., the number of entries in the vector $\hat{\mu}_\ma$). The noise introduced during the estimation process (such that $\hat{\mu}_\ma = \mu_\ma + \text{noise}$) is zero-mean, so the expected squared L2 error between the estimate $\hat{\mu}_\ma$ and the true marginal $\mu_\ma$ is:
\[ E\left[\|\hat{\mu}_\ma - \mu_\ma\|_2^2\right] = \sum_{j=1}^{n_\ma} \text{Var}\paren{(\hat{\mu}_\ma)_j} \]
Given that Theorem 4.7 asserts that $\text{Var}\paren{(\hat{\mu}_\ma)_j} = \sum_{\ra \subseteq \ma} \sigma_\ra^2 v_\ra$ for each entry $j$, the expected squared L2 error becomes:
\begin{equation} \label{eq:expected_error_derivation}
E\left[\|\hat{\mu}_\ma - \mu_\ma\|_2^2\right] = n_\ma \sum_{\ra \subseteq \ma} \sigma_\ra^2 v_\ra
\end{equation}
This is the result presented in \cref{prop:expected_error}. Since $n_\ma$ is a constant positive term for a given marginal $\ma$, it does not affect the optimal noise allocation in the CRP problem and can be ignored for optimization purposes.

\subsection{Solving CRP} \label{app:crp_solver}
Recall the conditional ResidualPlanner (CRP) problem is defined as:
\begin{equation} \label{eq:crp_problem_full}
\begin{aligned}
    \argmin_{\sigma_\ra^2} \quad & \sum_{\ra \subseteq \ma} \frac{v_\ra}{\frac{1}{\sigma_\ra^2} + \frac{1}{\tilde{\sigma}_\ra^2}} \\
    \text{s.t.} \quad & \sum_{\ra \subseteq \ma} \frac{p_\ra}{\sigma_\ra^2} \leq C \\
        & \sigma_\ra^2 > 0 \quad \forall \ra
\end{aligned}
\end{equation}
Here, $\sigma_\ra^2$ is the noise scale to be optimized for measuring the residual $\ra$ in the current round, $\tilde{\sigma}_\ra^2$ is the noise scale of our current estimate of $\ra$, $C = 2\rho$ is the privacy cost, with $\rho$ being the zCDP value allocated to this measurement, and $v_\ra, p_\ra$ are problem parameters defined in \cref{prop:crp_p} and \cref{prop:expected_error}.

\textbf{Stabilizing transformation}.
The original formulation \cref{eq:crp_problem_full} presents two potential issues for numerical optimization. Firstly, the privacy cost value $C$ can be very small, potentially leading to instability. Secondly, the objective function's dependence on $\frac{1}{\sigma_\ra^2}$ can create high curvature, also hindering optimization.
To address these, we introduce the transformed variables $x_\ra = \frac{1}{C\sigma_\ra^2}$ and $a_\ra = \frac{1}{C\tilde{\sigma}_\ra^2}$.
The original constraint $\sigma_\ra^2 > 0$ implies $x_\ra = \frac{1}{C\sigma_\ra^2} > 0$. However, in this setting, infinite variance ($\sigma_\ra^2 \to \infty$ and $x_\ra = 0$) corresponds to omitting the measurement of $\ra$, which is a valid outcome. Thus, the constraint on the transformed variable $x_\ra$ becomes $x_\ra \ge 0$.
Substituting these into the original problem leads to the transformed CRP problem:
\begin{equation} \label{eq:crp_problem_transformed}
\begin{aligned}
    \argmin_{x_\ra} \quad & \frac{1}{C} \sum_{\ra \subseteq \ma} \frac{v_\ra}{x_\ra + a_\ra} \\
    \text{s.t.} \quad & \sum_{\ra \subseteq \ma} x_\ra p_\ra \leq 1 \\
            & x_\ra \ge 0 \quad \forall \ra
\end{aligned}
\end{equation}
Notice that the leading constant $1/C$ in the objective function can be dropped without altering the minimizer $x_\ra$.

\textbf{Solving relaxed problem}.
Solving the transformed CRP problem \cref{eq:crp_problem_transformed} directly, which includes the non-negativity constraints $x_\ra \ge 0$ and an inequality budget constraint, can be challenging. Our strategy is to first address a simplified version. We define the relaxed CRP problem by omitting the non-negativity constraints $x_\ra \ge 0$ from \cref{eq:crp_problem_transformed} and by treating the budget constraint as an equality:
\begin{equation} \label{eq:crp_relaxed}
\begin{aligned}
    \argmin_{x_\ra} \quad & \sum_{\ra \subseteq \ma} \frac{v_\ra}{x_\ra + a_\ra} \\
    \text{s.t.} \quad & \sum_{\ra \subseteq \ma} x_\ra p_\ra = 1
\end{aligned}
\end{equation}
The budget constraint is formulated as an equality here because the objective function $\sum_{\ra \subseteq \ma} \frac{v_\ra}{x_\ra + a_\ra}$ (where $v_\ra > 0$ and $a_\ra > 0$) is monotonically decreasing with respect to each $x_\ra$. To minimize this sum, the values of $x_\ra$ should be made as large as the constraints permit. Therefore, an optimal solution to \cref{eq:crp_problem_transformed} must exhaust the budget, satisfying $\sum_{\ra \subseteq \ma} x_\ra p_\ra = 1$. This assumption simplifies the derivation of a closed-form solution for the relaxed problem.

\begin{theorem} \label{thm:relaxed_crp_solution}
Consider the relaxed CRP problem defined in \cref{eq:crp_relaxed}, which has the equality budget constraint $\sum_{\ra \subseteq \ma} p_\ra x_\ra = 1$ and no non-negativity constraints on $x_\ra$. The solution to this relaxed problem, denoted $x_\ra^{\text{relaxed}}$, is given by:
\begin{equation} \label{eq:relaxed_solution_theorem}
    x_\ra^{\text{relaxed}} = \frac{1 + Q}{S} \sqrt{\frac{v_\ra}{p_\ra}} - a_\ra
\end{equation}
where $S = \sum_{\ra \subseteq \ma} \sqrt{p_\ra v_\ra}$ and $Q = \sum_{\ra \subseteq \ma} p_\ra a_\ra$.
\end{theorem}

\begin{proof}
We use the method of Lagrange multipliers to find the solution $x_\ra^{\text{relaxed}}$ for the relaxed problem \cref{eq:crp_relaxed}, which features the equality constraint $\sum_{\ra \subseteq \ma} p_\ra x_\ra - 1 = 0$. The Lagrangian $L$ is:
\begin{equation}
    L(x, \lambda) = \sum_{\ra \subseteq \ma} \frac{v_\ra}{x_\ra + a_\ra} + \lambda \left( \sum_{\ra \subseteq \ma} p_\ra x_\ra - 1 \right)
\end{equation}
Here, $\lambda$ is the Lagrange multiplier. By construction, $v_\ra > 0$, $p_\ra > 0$, and $a_\ra > 0$. The objective function terms require $x_\ra + a_\ra \neq 0$; given $a_\ra > 0$, we anticipate $x_\ra + a_\ra > 0$.

To find $x_\ra^{\text{relaxed}}$, we take the partial derivative of $L$ with respect to each $x_\ra$ and set it to zero (the first-order condition for optimality):
\begin{equation}
    \frac{\partial L}{\partial x_\ra} = - \frac{v_\ra}{(x_\ra + a_\ra)^2} + \lambda p_\ra = 0
\end{equation}
Rearranging this equation yields:
\begin{align*}
    \lambda p_\ra &= \frac{v_\ra}{(x_\ra + a_\ra)^2} \\
    (x_\ra + a_\ra)^2 &= \frac{v_\ra}{\lambda p_\ra}
\end{align*}
Assuming $x_\ra + a_\ra > 0$, we take the positive square root:
\[ x_\ra + a_\ra = \sqrt{\frac{v_\ra}{\lambda p_\ra}} \]
This step implies $\lambda p_\ra > 0$. Since $p_\ra > 0$, we must have $\lambda > 0$. Solving for $x_\ra$ (which we denote $x_\ra^{\text{relaxed}}$ as it is the solution to the relaxed problem):
\begin{equation} \label{eq:x_ra_intermediate_form}
x_\ra^{\text{relaxed}} = \sqrt{\frac{v_\ra}{\lambda p_\ra}} - a_\ra
\end{equation}
Next, we apply the equality budget constraint $\sum_{\ra \subseteq \ma} p_\ra x_\ra^{\text{relaxed}} = 1$. Substituting the expression for $x_\ra^{\text{relaxed}}$ from \cref{eq:x_ra_intermediate_form} into this constraint gives:
\begin{equation}
    \sum_{\ra \subseteq \ma} p_\ra \left( \sqrt{\frac{v_\ra}{\lambda p_\ra}} - a_\ra \right) = 1
\end{equation}
Simplifying the sum:
\begin{align}
    \sum_{\ra \subseteq \ma} \left( \frac{1}{\sqrt{\lambda}} \sqrt{v_\ra p_\ra} - p_\ra a_\ra \right) &= 1 \\
    \frac{1}{\sqrt{\lambda}} \sum_{\ra \subseteq \ma} \sqrt{v_\ra p_\ra} - \sum_{\ra \subseteq \ma} p_\ra a_\ra &= 1
\end{align}
Let $S = \sum_{\ra \subseteq \ma} \sqrt{p_\ra v_\ra}$ and $Q = \sum_{\ra \subseteq \ma} p_\ra a_\ra$, as defined in \cref{thm:relaxed_crp_solution}. Substituting $S$ and $Q$ into the equation yields:
\begin{equation}
    \frac{S}{\sqrt{\lambda}} - Q = 1
\end{equation}
We now solve for $\sqrt{\lambda}$, noting that $S > 0$ and $1+Q > 0$ by construction:
\begin{align}
    \frac{S}{\sqrt{\lambda}} &= 1 + Q \\
    \sqrt{\lambda} &= \frac{S}{1 + Q}
\end{align}
Finally, we substitute this expression for $\sqrt{\lambda}$ back into the equation for $x_\ra^{\text{relaxed}}$ from \cref{eq:x_ra_intermediate_form}:
\begin{align}
    x_\ra^{\text{relaxed}} &= \frac{1}{\sqrt{\lambda}} \sqrt{\frac{v_\ra}{p_\ra}} - a_\ra \nonumber \\
        &= \left( \frac{1+Q}{S} \right) \sqrt{\frac{v_\ra}{p_\ra}} - a_\ra
\end{align}
This expression provides the solution $x_\ra^{\text{relaxed}}$ for the relaxed problem \cref{eq:crp_relaxed}, consistent with \cref{eq:relaxed_solution_theorem}.
\end{proof}

\textbf{Iterative solver}.
Let $x_\ra^{\text{relaxed}}$ be the solution to the relaxed problem as given by \cref{eq:relaxed_solution_theorem}. We check if this solution satisfies the non-negativity constraints $x_\ra^{\text{relaxed}} \ge 0$ for all $\ra$ in \cref{eq:crp_problem_transformed}. If all $x_\ra^{\text{relaxed}} \ge 0$, then $x_\ra^{\text{relaxed}}$ is feasible and thus optimal for \cref{eq:crp_problem_transformed}. In this case, we set the final solution $x_\ra^* = x_\ra^{\text{relaxed}}$.

However, if any $x_\ra^{\text{relaxed}} < 0$, this analytical solution is infeasible for \cref{eq:crp_problem_transformed}. The true optimum of \cref{eq:crp_problem_transformed} must then lie on the boundary where one or more $x_\ra=0$. In this scenario, we employ an iterative numerical solver. We use the Clarabel solver \citep{Clarabel_2024}, an interior-point solver for convex optimization, which we found to be fast and reliable for this problem that is available via CVXPY \citep{cvxpy}. We provide the (potentially infeasible) $x_\ra^{\text{relaxed}}$ values as a warm start to Clarabel to potentially speed up convergence, using a numerical tolerance of $10^{-8}$. We denote the solution to the full CRP problem obtained from this process as $x_\ra^*$.

\textbf{Post-processing}.
In some cases, the solution $x_\ra^*$ to the full CRP problem includes measurements of residuals $\ra$ with very high noise. Although theoretically optimal, this is impractical. We therefore omit such measurements. If a measurement's share of the privacy cost, $p_\ra x_\ra^*$, is below a threshold $\eta$, we do not perform that measurement. For the remaining residuals, the optimal measurement noise is recovered as $\sigma_\ra^2 = \frac{1}{Cx_\ra^*}$. In experiments, we use $\eta = 10^{-3}$.

\section{In-axis Running Time} \label{app:runtime}

We now analyze the running time of the in-axis algorithms.
\cref{thm:time_complexity} asserts that both $\De{\mu_\gamma}{\gamma}{\tau}$ and $\Rec{\zeta_\tau}{\tau}{\gamma}$ take $\mathcal O(|\gamma| n_\gamma)$ time where $n_\gamma = \prod_{k \in \gamma} n_k$ is the size of~$\mu_\gamma$.

Each algorithm is a sequence of $|\gamma|$ $n$-mode products on intermediate arrays of either increasing size ($\Recn$) or decreasing size ($\Den$).
The size of each intermediate array is always bounded by $n_\gamma$, the size of the marginal.
We claim that each $n$-mode product takes $\mathcal O(n_\gamma)$ time, which will immediately yield the result.
Write one of these $n$-mode products as $\mathcal U = \mathcal V \times_k A_k$ where $\mathcal V$ is the starting array and $\mathcal U$ is the resulting array: we claim that computing the $n$-mode product takes $\mathcal O(n_{\mathcal U} + n_{\mathcal V})$ time, where $n_{\mathcal U}$ and $n_{\mathcal V}$ are the sizes of $\mathcal U$ and $\mathcal V$, respectively.
In words, each entry of $\mathcal U$ and $\mathcal V$ is ``touched'' a constant number of times.
This follows because the computation can be divided into the operations $u = \texttt{op}(v)$ for each mode-$k$ fiber $v$ of $\mathcal V$, where \texttt{op} is one of the operations from \cref{tab:axis_ops}, and it is straightforward to verify that each of those operations takes $\mathcal O(n_u + n_v)$ time, where $n_u$ and $n_v$ are the sizes $u$ and $v$, respectively.
Finally because $n_{\mathcal U}$ and $n_{\mathcal V}$ are both bounded by $n_\gamma$, the time for each $n$-mode product is $\mathcal O(n_\gamma)$.

It is also possible to establish these running times is by relation to the implementation using Kronecker matrices and the ``shuffle algorithm'' from~\cite{grem} and~\cite{resplanner}.
The shuffle algorithm multiplies a Kronecker-product matrix $A_1 \otimes \cdots \otimes A_d$ by a vector $u$ via a sequence of operations that starts with the tensor $\ten(u)$ and then repeatedly reshapes the tensor into a matrix and multiplies by one factor $A_k$. Each reshape-and-multiply operation implements one $n$-mode product and has the same time complexity as multiplying each $k$-mode fiber by $A_k$. 
The overall time complexity depends on the efficiency of the matrix-vector multiplications with the factors $A_k$.
For the factors $A_k \in \R^{a \times b}$ that arise when converting between marginals and residuals, $A_k v$ can always be computed in $\mathcal O(a+b)$ time.
With this fact, the shuffle implementation was shown in~\citet{grem} and \citet{resplanner} to take $\mathcal O(|\gamma| n_\gamma)$ overall. 
Because our implementations perform the same sequences of $n$-mode products, they also take $\mathcal O(|\gamma| n_\gamma)$ time.

\section{AIM+GReM} \label{app:aim_grem}

In this section, we describe components of AIM+GReM and AIM+PGM not fully described in the main text. 
Let us first describe the initialization step of AIM+GReM, shown in Algorithm~\ref{alg:initial_measurements}, which measures residuals equivalent to measuring 1-way marginals with variance $\sigma_0^2$.
For each attribute $i$, we measure both the residual $z_{\{i\}}$ and the total query $z_\emptyset'$.
Since we measure the total query multiple times, we update the total query $z_\emptyset$ and its precision $\lambda_\emptyset$ via inverse-variance weighting each round.

\begin{algorithm}[ht]
    \caption{Initialization}
    \label{alg:initial_measurements}
    \begin{algorithmic}[1]
      \Require Marginals $\mu_0, \mu_{\{ 1 \}}, \ldots, \mu_{\{ d \}}$, variance $\sigma_0^2$
      \Ensure Residuals $z_0, z_{\{ 1 \}}, \ldots, z_{\{ d \}}$, precisions $\lambda_0, \lambda_{\{ 1 \}}, \ldots, \lambda_{\{ d \}}$
      \For{$i=1,\ldots,d$} 
        \State $z_{\{i\}} = \measure(\mu_{\{i\}}, \sigma_0^2)$, $\lambda_{\{i\}} = \sigma_0^{-2}$ \Comment{Measure 1-way residual}
        \State $z_\emptyset' = \measure(\mu_\emptyset, n_i \sigma_0^2)$ \Comment{Measure total query}
        \State $z_\emptyset = \frac{\lambda_\emptyset z_\emptyset + n_i^{-1} \sigma_0^{-2} z_\emptyset'}{\lambda_0 + n_i^{-1} \sigma_0^{-2}}$, $\lambda_\emptyset = \lambda_\emptyset + n_i^{-1} \sigma_0^{-2}$ \Comment{Update total query}
      \EndFor
    \end{algorithmic}
\end{algorithm}

This differs from AIM+PGM in two ways: AIM+PGM measures all 1-way marginals without decomposing them into residuals and uses a different noise scale.
In line 1 of Alg.~\ref{alg:aimGReM}, AIM+GReM initializes the noise scale at $\sigma_0^2 = \frac{| \W^\downarrow |}{0.9 \rho}$. 
In contrast, AIM+PGM initializes the noise scale at $\sigma_0^2 = \frac{8d}{0.9 \rho}$, where $d$ is the number of attributes in the data domain.
The number of possible measurements AIM+GReM makes scales with the size of the workload $\W^\downarrow$, while AIM+PGM scales with the number of attributes in the data domain.

\subsection{Shared AIM Components} \label{app:aim}

AIM+GReM and AIM+PGM share two main algorithmic components: the selection step and the budget annealing step.
We describe these components in detail below.

\textbf{Selection step}.
The selection step of AIM+PGM uses the exponential mechanism to select a marginal to measure each round.
For $\gamma \in \W^\downarrow$ at iteration $t$, the score function is given by 
\begin{equation*}
  \text{Score}_\gamma = w_\gamma \left( \| \mu_\gamma - \hat \mu_\gamma \| - \sqrt{2/\pi} \cdot \sigma_t \cdot n_\gamma  \right)
\end{equation*}
where $w_\gamma = \sum_{\pi \in \W} | \gamma \cap \pi | $.
This can be interpreted as the improvement in $L_1$ error we can expect in the reconstructed answers by measuring a given marginal, weighted to prioritize lower-order marginals whose attributes are contained in higher-order marginals in $\W^\downarrow$. 

\textbf{Budget annealing}.
The annealing step, given in Alg.~\ref{alg:annealing}, updates the selection privacy parameter $\epsilon_{t+1}$ and the measurement noise scale $\sigma_{t+1}^2$ and for the next round. 
This approach checks if the change between the current and prior reconstructed answer for the selected marginal is greater or less than the expected improvement. 
If the difference is less, the mechanism doubles the privacy budget spent on both selection and measurement for the next round (lines 1-6). 
This allows the mechanism to adapt both the selection step and measurement step to the privacy budget. 
If the mechanism is close to exhausting the privacy budget and cannot run for two additional rounds, the mechanism spends the remaining budget on the next round (lines 7-9).

\begin{algorithm}[ht]
  \caption{Budget Annealing}
  \label{alg:annealing}
  \begin{algorithmic}[1]
    \Ensure Selection budget $\epsilon_{t+1}$, measurement noise scale $\sigma_{t+1}^2$
    \State \textbf{if} $\| \mu_\gamma - \hat \mu_\gamma \| \leq \sqrt{2/\pi} \cdot \sigma_t \cdot n_\gamma$ \textbf{then do} \Comment{Checks the annealing condition}
    \State \qquad $\epsilon_{t+1} = 2 \epsilon_t$ \Comment{Doubles the selection budget}
    \State \qquad $\sigma_{t+1}^2 = \nicefrac{\sigma_t^2}{4}$ \Comment{Doubles the measurement budget}
    \State \textbf{else do}
    \State \qquad $\epsilon_{t+1} = \epsilon_t$
    \State \qquad $\sigma_{t+1}^2 = \sigma_t^2$
    \State \textbf{if} $(\rho - \rho_{used}) \leq 2(\nicefrac{1}{2\sigma_{t+1}^2} + \nicefrac{\epsilon_{t+1}^2}{8})$ \textbf{then do} \Comment{Checks if at least two rounds can be run}
    \State \qquad $\epsilon_{t+1} = \sqrt{0.8 \cdot (\rho - \rho_{used})}$
    \State \qquad $\sigma_{t+1}^2 = \nicefrac{1}{1.8 (\rho - \rho_{used})}$
  \end{algorithmic}
\end{algorithm}

\subsection{Privacy Analysis} \label{app:aim_grem_priv}

\privThm*

\begin{proof}
    The initialization step of AIM+GReM makes $d$ 1-way marginal measurements with noise scale $\sigma_0^2 = \frac{\abs{\W^\downarrow}}{0.9 \rho}$, spending budget $\rho_{init} = \frac{1}{2\sigma_0^2} = \frac{0.9 \rho}{2\abs{\W^\downarrow}}$ by Proposition \ref{defn:gaussian mech}. 
    By Proposition \ref{prop:zcdpProperties}(1), the initialization step satisfies $d \rho_{init}$-zCDP.
    Observe that $\rho_{init} < \rho$, since $d < \abs{\W^\downarrow}$.
    This shows that the privacy budget is not overspent during initialization.
    Let $t > d$ denote the round of the mechanism. 
    Each round $t$ of AIM+GReM calls the exponential mechanism with privacy parameter $\epsilon_t$ and the Gaussian mechanism with noise scales output by solving a CRP problem with budget $\frac{1}{2 \sigma_t^2}$.
    By Propositions \ref{defn:exponential mech}, \ref{prop:crp_p}, \ref{prop:zcdpProperties}(1), round $t$ satisfies $\rho_t$-zCDP, where $\rho_t = \frac{\epsilon_t^2}{8} + \frac{1}{2\sigma_t^2} $.
    The cumulative privacy budget spent at round $t$ is $\rho_{used} = \rho_{init} + \sum_{i=d+1}^t \rho_i$.
    It remains to show that $\rho_{used}$ never exceeds $\rho$.
    The annealing step (Alg.~\ref{alg:annealing}) checks the condition $(\rho - \rho_{used}) \leq 2(\rho_{t+1})$.
    If the condition is not met, then $(\rho - \rho_{used}) > 2(\rho_{t+1})$, so round $t+1$ can run without overspending the privacy budget.
    If the condition is met, then the mechanism spends the remaining privacy budget on round $t+1$ so that $\rho_{t+1} = \frac{\left(\sqrt{0.8 (\rho - \rho_{used})}\right)^2}{8} + \frac{1}{\left( \frac{2}{1.8(\rho - \rho_{used})} \right)} = \frac{\rho - \rho_{used}}{10} + \frac{9(\rho - \rho_{used})}{10} = \rho - \rho_{used}$.
\end{proof}

\section{Experiments} \label{app:experiments}

\subsection{Experimental Details} \label{app:exp_details}

\textbf{Compute environment}. All experiments were run on an internal CPU compute cluster. Experiments were allocated 20GB of memory and two cores with a limit of at most 14 days running time.

\textbf{Datasets}.
Datasets included in the evaluation are described in Table~\ref{table:datasets}.
Note that we process the datasets by dropping any attributes with fewer than two unique values, dropping any records with missing values, and subset the datasets such that no attribute has more than 100 unique values. From these processed datasets, we infer data domains. 

\begin{table}[h!]
   \centering
   \begin{tabular}{c|c|c|c|c}
   \textbf{Dataset} & \textbf{Records} & \textbf{Attributes} & \makecell{\textbf{Average} \\ \textbf{Attribute Size}} & \makecell{\textbf{Total} \\ \textbf{Domain Size}} \\ \hline 
   Adult & 48,842  & 15 & 16.4  & $9.44*10^{15}$ \\
   ACS   & 378,817 & 66 & 4.22  & $3.82*10^{18}$ \\ 
   Loans & 42,535  & 43 & 50.01 & $1.54*10^{63}$
   \end{tabular}
   \caption{Summary of datasets used in the experiments.}
   \label{table:datasets}
\end{table}

\subsection{Additional Experiments} \label{app:additional_experiments}

In this section, we provide additional experimental results. In Figure~\ref{fig:speed_linear}, we report results from Figure~\ref{fig:speed} comparing our in-axis approach with the baseline Shuffle algorithm for the $\Den$ and $\Recn$ operations.
\begin{figure*}
    \centering
    \includegraphics[width=1\textwidth]{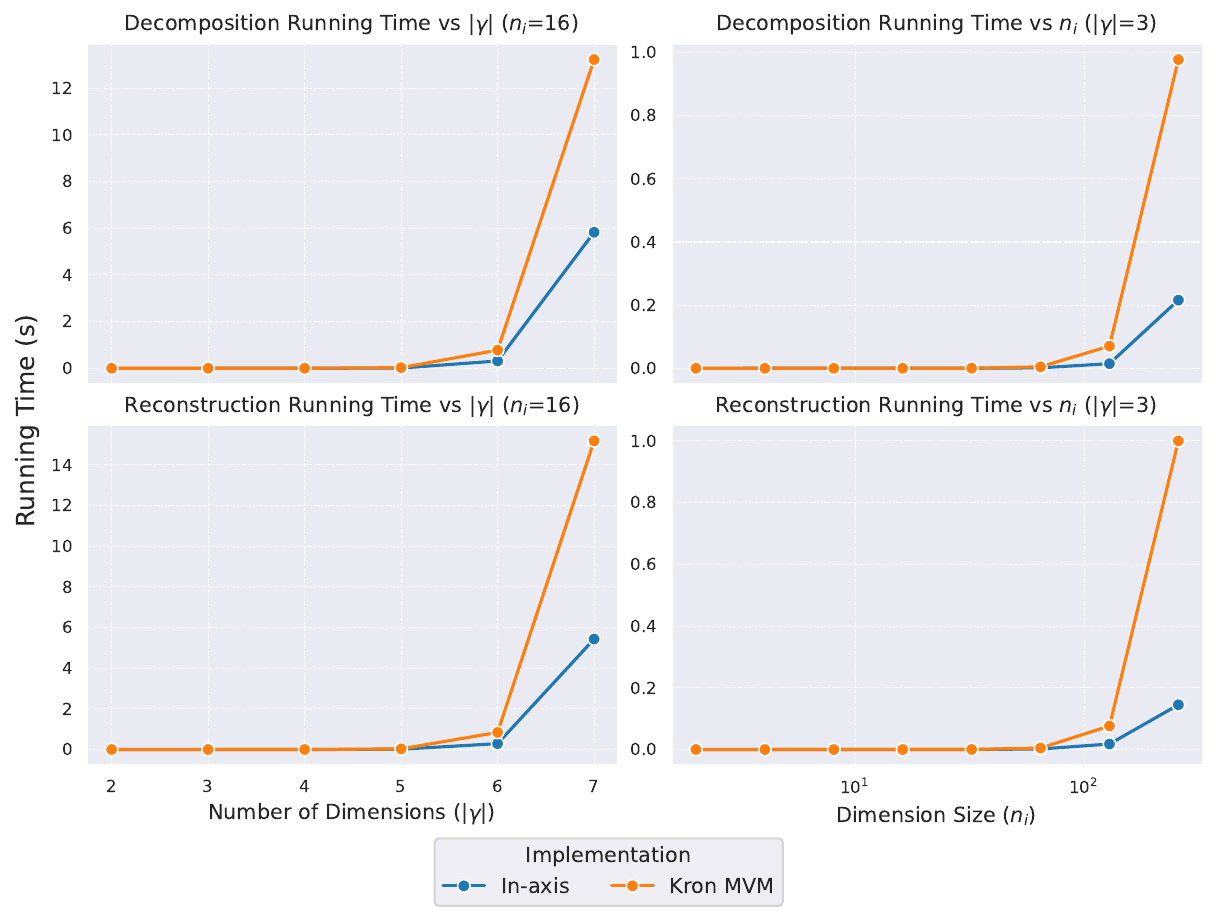}
    \caption{Running time in seconds of the $\Den$ and $\Recn$ operations with our in-axis implementation versus the baseline Kronecker-based approach.}
    \label{fig:speed_linear}
\end{figure*}
The remaining plots compare AIM+GReM to AIM+PGM and ResidualPlanner.
First, we display the results of Figure~\ref{fig:aimGrem_l1_by_time} in absolute terms. 
In Figure~\ref{fig:aimGrem_l1_error}, we report we report the mean $L_1$ error of the three methods by the privacy budget $\epsilon$. 
In Figure~\ref{fig:aimGrem_l1_time}, we report we report the mean running time in seconds by the privacy budget $\epsilon$. 
In Figure~\ref{fig:aimGrem_l1_time_by_error}, we report the mean running time in seconds by mean $L_1$ error.
\begin{figure}[ht] 
    \centering
    \includegraphics[width=\textwidth]{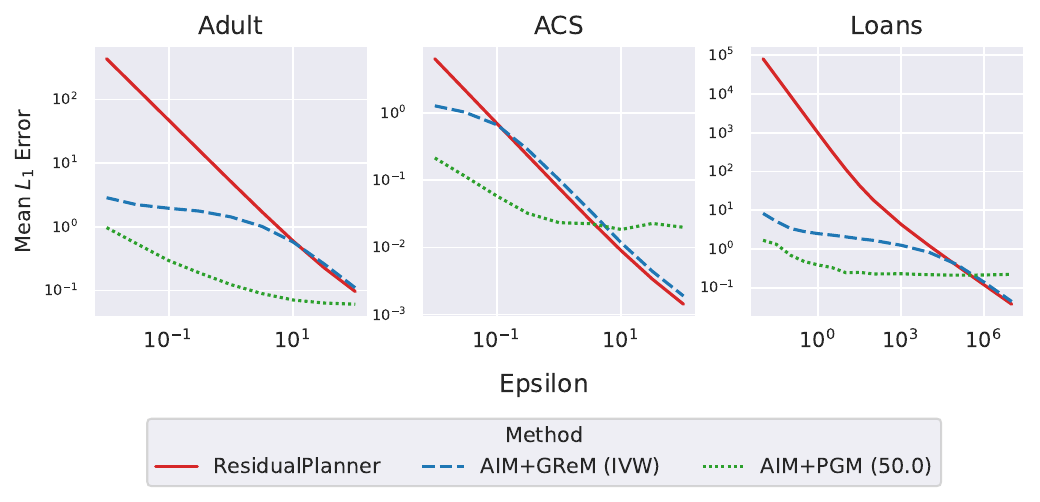}
    \caption{\small Mean $L_1$ error of AIM variants by $\epsilon$ on all 3-way marginals. Results are averaged across five trials. Privacy budgets $\epsilon$ indicated on the horizontal axis and $\delta = 10^{-9}$. AIM+PGM results use model size 50MB.}   
    \label{fig:aimGrem_l1_error}
\end{figure}
\begin{figure}[ht]
    \centering
    \includegraphics[width=\textwidth]{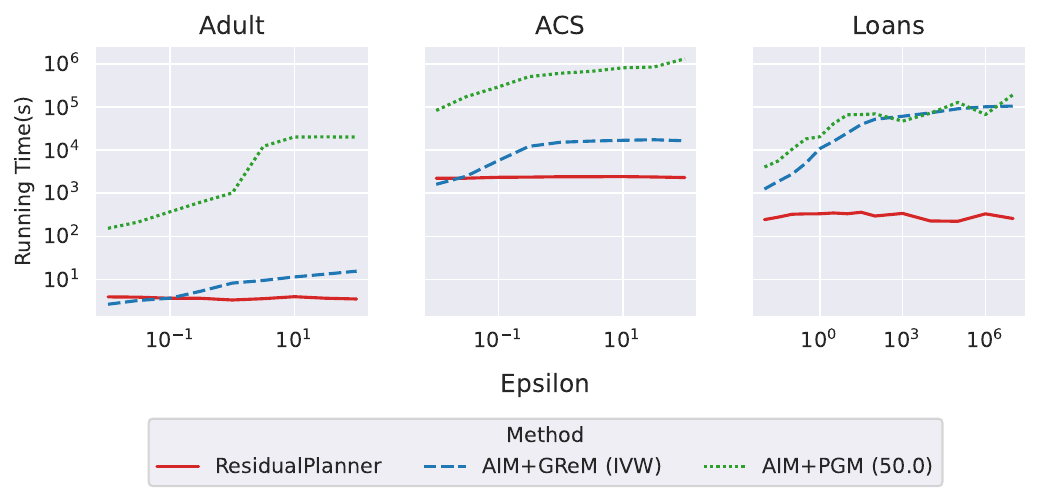}
    \caption{\small Running time of AIM variants in seconds by $\epsilon$ on all 3-way marginals. Results are averaged across five trials. Privacy budgets $\epsilon$ indicated on the horizontal axis and $\delta = 10^{-9}$. AIM+PGM results use model size 50MB.}   
    \label{fig:aimGrem_l1_time}
\end{figure}
\begin{figure}[ht]
    \centering
    \includegraphics[width=\textwidth]{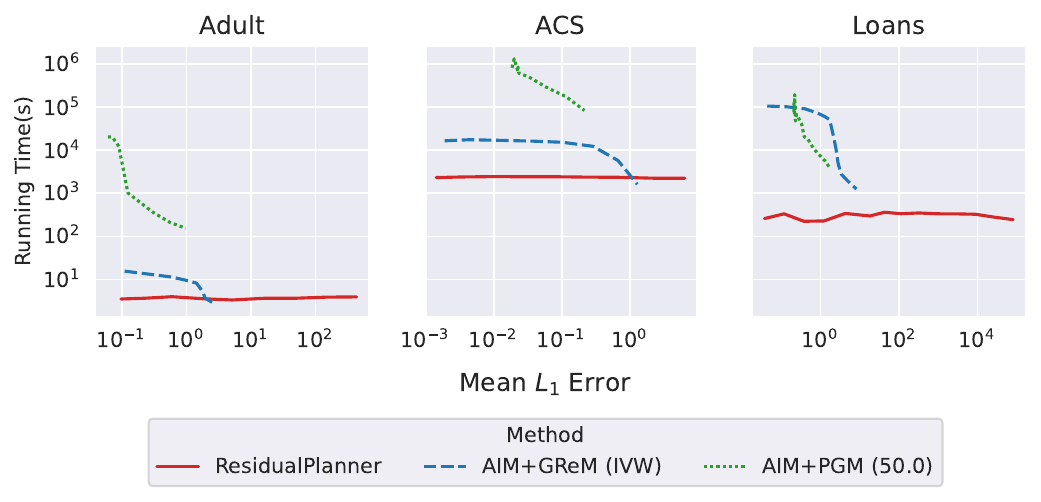}
    \caption{\small Running time of AIM variants in seconds by mean $L_1$ error on all 3-way marginals. Results are averaged across five trials. AIM+PGM results use model size 50MB.}   
    \label{fig:aimGrem_l1_time_by_error}
\end{figure}
Next, we display the mean $L_2$ and max $L_1$ error of the three methods on the task of reconstructing all 3-way marginals by running time, stratified by privacy budget, in Figures \ref{fig:aimGrem_l2_by_time} and \ref{fig:aimGrem_max_by_time}, respectively.
These results are consistent with the findings in Figure~\ref{fig:aimGrem_l1_by_time}, presented in the main text.
Figure~\ref{fig:aim_L1_by_time} shows results for AIM+PGM with varying model size from 1MB to 100MB by running time.
We observe that the reported results are robust to the model size.

\begin{figure}[ht]
    \centering
    \includegraphics[width=\textwidth]{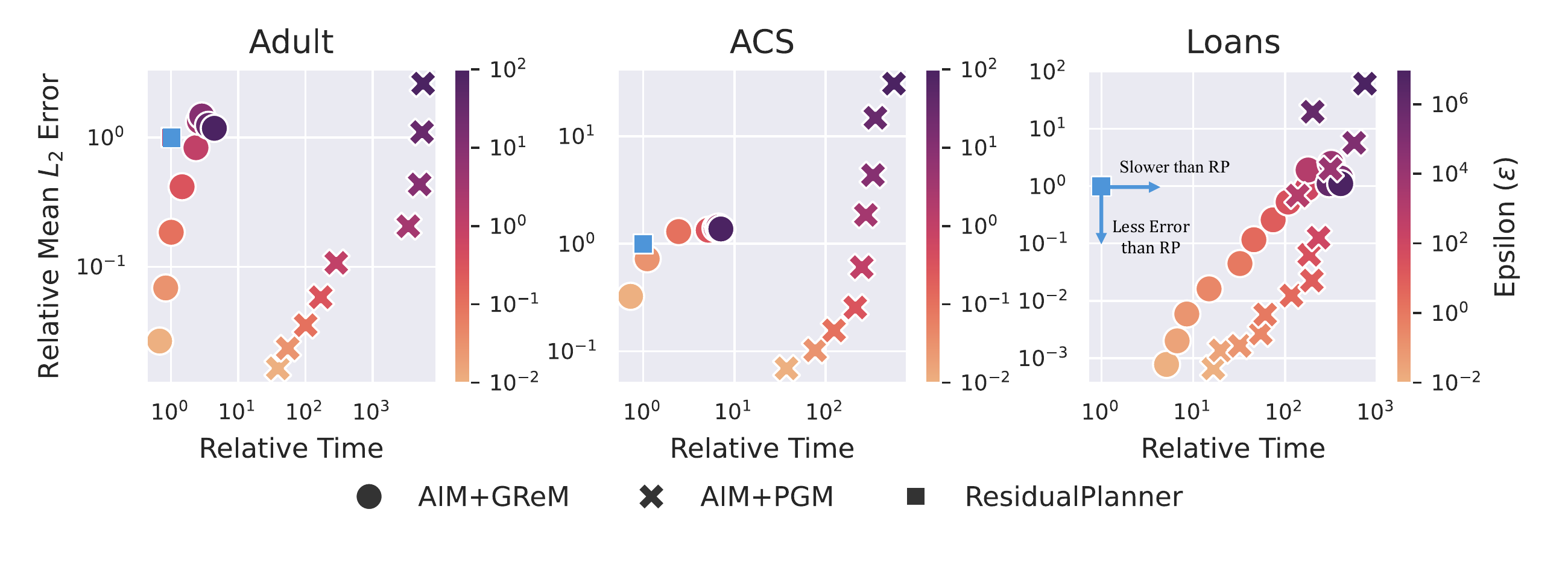}
    \caption{\small Mean $L_2$ error of AIM variants by running time on all 3-way marginals. Quantities are given relative to ResidualPlanner (i.e. AIM error / ResidualPlanner error). Results are averaged across five trials. Privacy budgets $\epsilon$ indicated in the vertical scale and $\delta = 10^{-9}$. AIM+PGM results use model size 50MB.}   
    \label{fig:aimGrem_l2_by_time}
\end{figure}

\begin{figure}[ht]
    \centering
    \includegraphics[width=\textwidth]{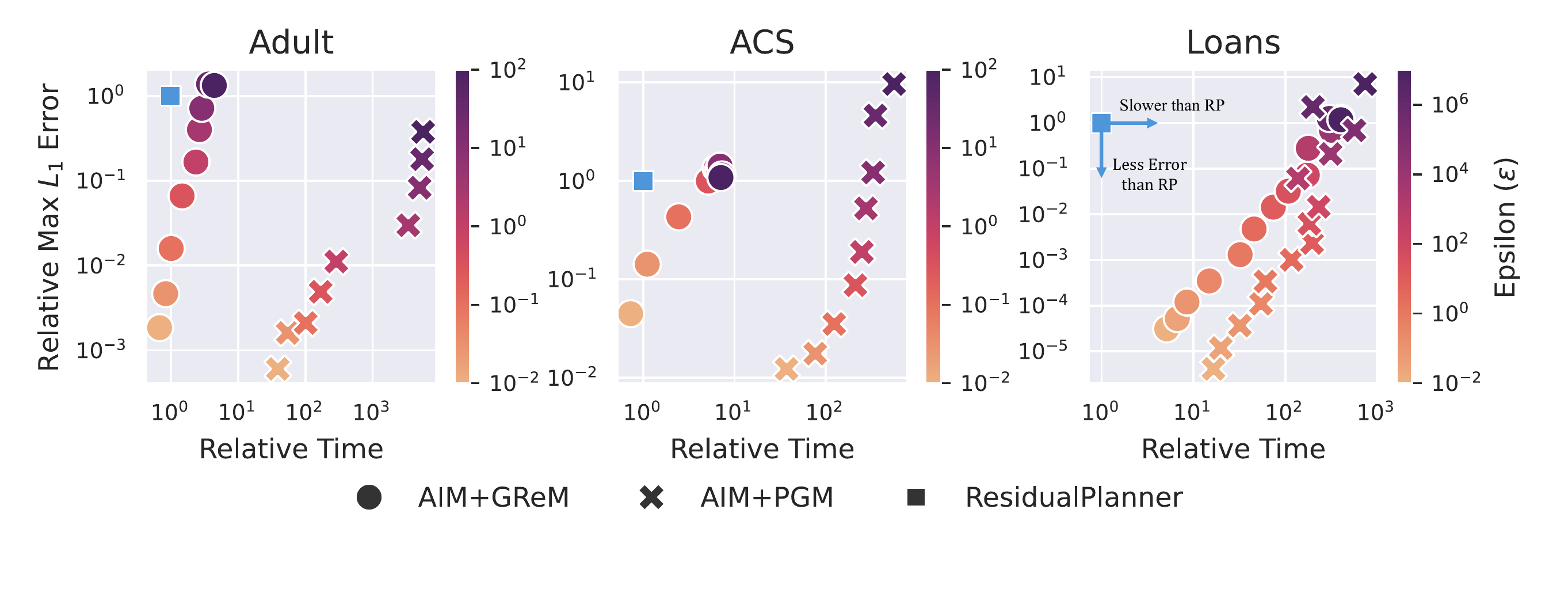}
    \caption{\small Max $L_1$ error of AIM variants by running time on all 3-way marginals. Quantities are given relative to ResidualPlanner (i.e. AIM error / ResidualPlanner error). Results are averaged across five trials. Privacy budgets $\epsilon$ indicated in the vertical scale and $\delta = 10^{-9}$. AIM+PGM results use model size 50MB.}   
    \label{fig:aimGrem_max_by_time}
\end{figure}

\begin{figure}[ht]
    \centering
    \includegraphics[width=\textwidth]{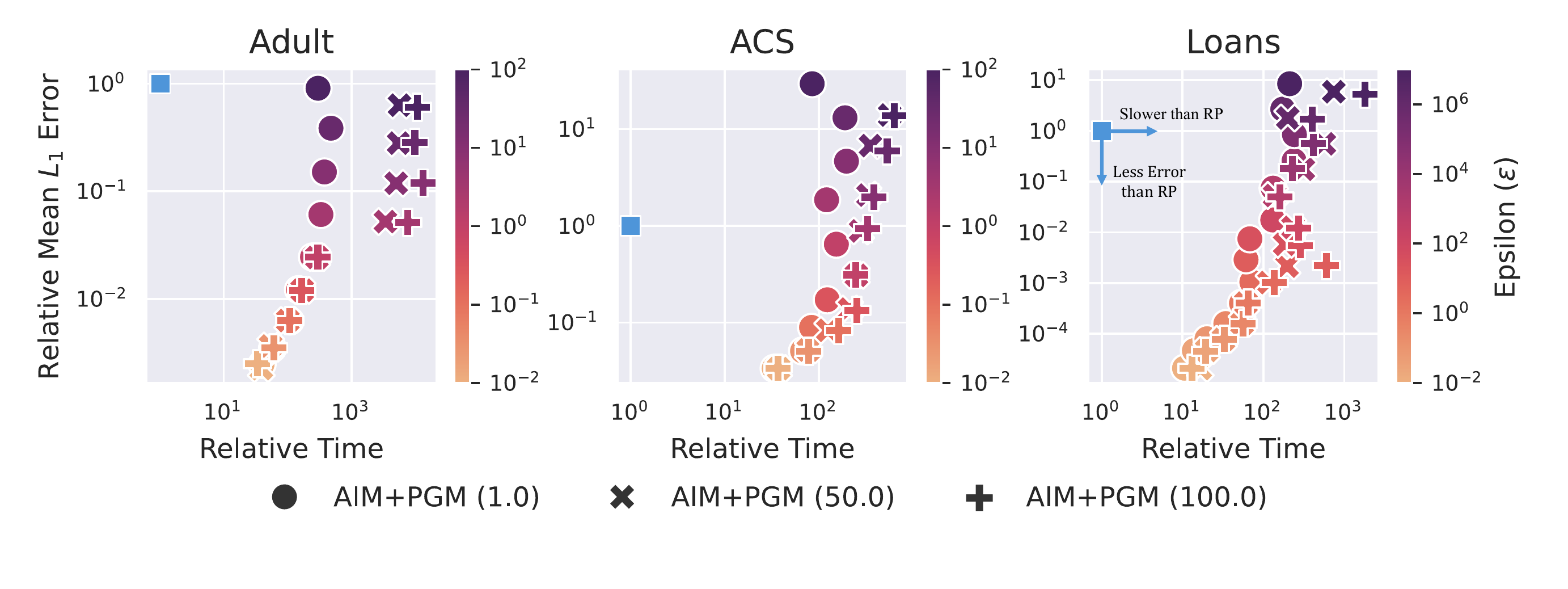}
    \caption{\small Mean $L_1$ error of AIM+PGM with varying model size by running time on all 3-way marginals. Quantities are given relative to ResidualPlanner (i.e. AIM error / ResidualPlanner error). Results are averaged across five trials. Privacy budgets $\epsilon$ indicated in the vertical scale and $\delta = 10^{-9}$.}   
    \label{fig:aim_L1_by_time}
\end{figure}

\end{document}